\documentclass[12pt,onecolumn,english]{IEEEtran}
\usepackage{tablists}
\usepackage{amsmath}
\usepackage{amsthm}
\usepackage{amssymb}
\usepackage{stmaryrd}
\usepackage{tikz}
\usepackage{pgfplots}
\usepackage{subcaption}
\usepackage{enumerate}
\usepackage{blkarray}
\usepackage{hyperref}
\hypersetup{
    colorlinks,
    linktocpage,
    linkcolor=blue,
}
\usepackage{romannum}

\newcommand{\markov}{\mathrel\multimap\joinrel\mathrel-\mspace{-9mu}\joinrel\mathrel-}
\newtheorem{theorem}{Theorem}
\newtheorem{definition}{Definition}
\newtheorem{lemma}{Lemma}
\newtheorem{corollary}{Corollary}
\newtheorem{proposition}{Proposition}
\newtheorem{example}{Example} 
\newtheorem{remark}{Remark}

\title{Capacity Results for the K-User Broadcast Channel with Two Nested Multicast Messages}

\author{Mohamed Salman and Mahesh K. Varanasi \\
\thanks{This work will be presented in part at the 2017  International Symposium on Information Theory (ISIT) in Aachen, Germany. 

This research was funded in part by Gift 24868 from Qualcomm Inc,. San Diego, CA, and by NSF Grant 1423657. 

M. Salman and M. K. Varanasi are with the Electrical, Computer and Energy Engineering Department, University of Colorado, Boulder, CO, USA (emails: \{mohamed.salman, varanasi\}@colorado.edu).}
}

\begin{document}

\maketitle
\thispagestyle{empty}
\pagestyle{empty}

\begin{abstract}
The $K$-user discrete memoryless (DM) broadcast channel (BC) with two nested multicast messages is studied in which one common message is to be multicast to all receivers and the second private message to a subset of receivers. The receivers that must decode both messages are referred to as private receivers and the others that must decode only the common message as common receivers. For two nested multicast messages, we establish the capacity region for several classes of DM BCs characterized by the respective associated sets of pair-wise relationships between and among the common and private receivers, each described by the well-known more capable or less noisy conditions.

For three classes of DM BCs, the capacity region is simply achieved by superposition coding and the proofs of the converses rely on a recently found information inequality. 
The achievable rate region is then enhanced through the addition of a splitting of the private message into as many parts as there are common receivers and indirect decoding. A closed-form two-dimensional polyhedral description is obtained for it for a given coding distribution. Through a converse result that relies on the well-known Csiszar sum lemma and the information inequality, a specialization of this region that involves splitting the private message into just two sub-messages is proved to be the capacity region for several classes of DM BCs, beyond those for which superposition coding alone is capacity optimal, thereby underscoring the benefit of rate-splitting. 

All previously known capacity results for DM BCs with two nested multicast messages for the two and three-receiver DM BCs as well as DM BCs with one private or one common receiver are included in the general framework presented in this work.



\end{abstract}

\section{Introduction}

The discrete-memoryless (DM) broadcast channel (BC) is an archetypal communication setting in network information theory where a transmitter is to send messages to $K$ receivers. 
In this work, our focus is on the DM BC with two nested multicast messages, also known as degraded messages, wherein the transmitter must multicast a common message to all receivers, and a private message to a subset of receivers. 
This scenario is relevant in applications such as video and music broadcasting to many receivers which require two different levels of quality. The common message represents the standard quality multicast information that must be decodable by all receivers, while the private multicast message is a refinement to be decoded by a subset of receivers that, aggregated with the common message, represents the high quality message.

In spite of a great deal of attention starting in the 1970's, the capacity region of the BC with two nested multicast messages is known to date only for certain combination of numbers of users and classes of channels. We recall the definition of the well-known orderings between pairs of channels taken from the set of $K$ point-to-point channels associated with the transmitter and each of the $K$ receivers (which in turn can be used to define classes of BCs). Denote the transmitted symbol by the random variable $X$, an auxiliary random variable by $U$, and the received symbols at the $K$ receivers by $Y_1, Y_2 , \cdots , Y_K$, which are also used to refer to the respective receivers. The (conditional) probability mass functions of these random variables are denoted as is standard in the literature.
\begin{definition} \label{Def_degraded} Receiver $Y_c$ is said to be physically degraded from $Y_s$ if $p(y_s,y_c|x){=}p(y_s|x)p(y_c|y_s)$, i.e., $X {\markov} Y_s {\markov} Y_c$ forms a Markov chain. 
\end{definition} 
\begin{definition}  \label{Def_Less_Noisy} \cite[Definition 2]{korner1975source}
Receiver $Y_s$ is less noisy than $Y_c$ if $I(U;Y_s){\geq} I(U;Y_c)$ for all $p(u,x)$. Henceforth, we denote this condition as $Y_s {\succ} Y_c$.
\end{definition}
\begin{definition} \cite[Definition 3]{korner1975source} \label{Def_More_Capable}
Receiver $Y_s$ is more capable than $Y_c$ if $I(X;Y_s){\geq} I(X;Y_c)$ for all $p(x)$. Henceforth, we denote this condition as $ Y_s \succeq Y_c $. 
\end{definition} 
The less noisy condition is strictly less restrictive than the degraded relation (cf. \cite[Counterexample 1]{korner1975source}) and the more capable condition is in turn strictly less restrictive than the less noisy condition. 

When convenient, we refer to the less noisy condition $Y_s {\succ} Y_c$ as $Y_c$ being more noisy than $Y_s$ or denote it equivalently as $Y_c  \prec  Y_s$. Similarly, we may refer to the more capable condition $ Y_s \succeq Y_c $ as $Y_c$ being less capable than $Y_s$ or denote it equivalently as $Y_c \preceq Y_s$.

The two-receiver DM BC with degraded messages was first studied by Korner and Marton in 1977 \cite{kiirner1977general}. Indeed, they found the capacity region for the general case without any restrictions on the channel to be the set of rate pairs ($R_0$ and $R_1$, denoting the rates of the common and private messages, respectively) satisfying
\begin{align}
R_0&\leq I(U;Y_2)\nonumber \\
R_1&\leq I(X;Y_1|U)\nonumber \\
R_0+R_1&\leq I(X;Y_1) \label{Region_KM_2}
\end{align} 
for some $p(u,x)$.
The achievability scheme uses superposition coding and joint decoding. The common message, denoted henceforth as $M_0$, is represented by the auxiliary random variable $U$ with associated cloud center codewords, and the private message, denoted henceforth  as $M_1$, is superposed on $U$ to generate the satellite codewords (according to $p(x|u)$) to be transmitted, denoted by $X$. Receiver $Y_2$ finds $M_0$ by decoding $U$ and $Y_1$ finds ($M_0,M_1$) by decoding $X$. This scheme was shown earlier to achieve the capacity of the degraded DM BC $X{\markov}Y_1{\markov}Y_2$ with private messages by Gallager \cite{gallager1974capacity}. The main contribution in \cite{kiirner1977general} was thus to provide the converse proof using the images-of-a-set technique \cite{korner1977images}. 

More recently, Diggavi and Tse addressed the DM BC with arbitrary $K$ \cite{diggavi2006opportunistic} with two general nested multicast messages, $M_0$ and $M_1$, in which $M_0$ is to be decoded by all receivers and $M_1$ by a subset $\mathcal{S}\subseteq \{1,2,...,K\}$ of private receivers, which is also the setting studied in this paper. Henceforth, we will let $\mathcal{S} = \{1,2,...,L\}$ without loss of generality, so that the first $L$ receivers are private receivers and the last $K-L$ receivers are common receivers.  Using the same encoding scheme as in the 2-user case, and with each private receiver successively decoding the two messages (with common message first), a natural extension of the Korner-Marton region to this more general setting is the set of rate pairs 
\begin{align}
R_0&\leq I(U;Y_i) \enspace i\in \{1,2,...,K\} \nonumber \\
R_1&\leq I(X;Y_s|U) \enspace s\in \mathcal{S}
\label{Region_Ext_Marton}
\end{align}
for some $p(u,x)$. The authors of \cite{diggavi2006opportunistic} proved that the above region is optimal for certain classes of DM BCs respectively in two cases, namely, the single private receiver case of $\mathcal{S}{=}\{1\}$ and the single common receiver setting of  $\mathcal{S}{=}\{1,  \cdots, K-1 \}$. The class of DM BCs for which \eqref{Region_Ext_Marton} was shown to be the capacity region when $\mathcal{S}{=}\{1\}$ is defined by the $K-1$ Markov conditions $X {\markov} Y_1 {\markov} Y_i$ for each $i{\in}\{2,3,..,K\}$  and when $\mathcal{S}=\{1,  \cdots, K{-}1 \}$ by the $K-1$ Markov conditions $X{\markov} Y_i {\markov} Y_K$ for each $i {\in} \{ 1,2,..,K{-}1\}$. In both cases, the capacity regions are thus known for the respective classes of DM BCs defined by degradedness conditions, with each common receiver being a degraded version of each private receiver.

The problem of characterizing capacity for the DM BC with two nested multicast messages for a larger class of DM BCs but for the three-receiver case was addressed by Nair and El Gamal \cite{nair2009capacity} for the above two cases. When $\mathcal{S}=\{1\}$, they showed that the region (\ref{Region_Ext_Marton}) is not in general optimal \cite[Section IV]{nair2009capacity} and proposed an achievability scheme that involves rate splitting with superposition coding and indirect decoding. $M_0$ is represented by $U$ (the cloud centers), a part of $M_1$ is superposed on $U$ to obtain $V$, and the rest of $M_1$ is superposed on $V$ to yield $X$. Receiver $Y_3$ finds $M_0$ by decoding $U$ and $ Y_1$ finds $M_0$  and $M_1$ by decoding $X$ using joint typicality decoding, while $Y_2$ finds $M_0$ {\em indirectly} by decoding $V$. Using this scheme, the authors of \cite{nair2009capacity} established the capacity region for the class of channels constrained by a single less noisy condition $Y_3 \prec  Y_1$, thereby significantly enlarging the class of BCs for which capacity was known for $K{=}3$ earlier due to \cite{diggavi2006opportunistic}.
 
The $K=3$ and $\mathcal{S}=\{1,2\}$ case was also addressed in \cite{nair2009capacity} for which the authors showed therein that the natural extension of the Korner-Marton region given by (\ref{Region_Ext_Marton}) is the capacity region for the class of channels defined by two restrictions, $Y_3 \prec Y_1$ {\em and} $Y_3 \prec Y_2$ \cite[Proposition 11]{nair2009capacity}\footnote{There is a typo in the statement of Proposition 11 of \cite{nair2009capacity}. The conditions stated therein $Y_2 \prec  Y_1 $ and $Y_2 \prec  Y_3 $ should be $Y_3 \prec  Y_2 $ and $Y_3 \prec   Y_1$ (and these are stated correctly in the proof of Proposition 11 \cite{nair2009capacity}).} in thereby strictly expanding the class of DM BCs for which capacity was known from \cite{diggavi2006opportunistic}.
The converse proofs of \cite{nair2009capacity} for the single private and single common receiver cases use mainly the Csiszar sum lemma \cite[Lemma 7]{csiszar1978broadcast}.

The problem of transmitting two nested multicast messages has also been addressed in the more restricted context of combination networks \cite{ngai2004network}. These networks lie at the intersection of the multi-hop wired networks and single-hop BCs and can be considered as an example of deterministic DM BCs. In \cite{bidokhti2016capacity}, Bidokhti {\em et al} showed that a form of linear superposition coding with rate-splitting is optimal for two common receivers and any number of private receivers, i.e, $K=L+2$. When the number of common receivers increases beyond 2 ($K>L+2$), the linear superposition coding with rate splitting is shown to be not optimal anymore. For more than two common receivers, a new achievable rate region, which depends on block Markov encoding scheme, is obtained. This new inner bound is shown to be the capacity region for networks with three (or fewer) common receivers and any number of private receivers. Finally, a new inner bound for the DM BC with two nested multicast messages is proposed by extending the block Markov coding scheme to this case. The work of \cite{bidokhti2016capacity} however does not preclude the possibility that rate-splitting and superposition coding for the DM BC with two nested multicast messages with three common receivers when specialized to the combination network may suffice to obtain the capacity region, obviating the need for the block Markov coding scheme proposed therein.

In \cite{Romero-Varanasi:isit2016} and \cite{Romero-Varanasi:isit2017}, Romero and Varanasi, using notions from order theory and lattices, provide a general framework applicable for the $K$-user DM BC with any subset of the set of exponentially many possible messages, one for each subset of receivers. Therein, the authors obtain general inner bounds for rate-splitting and superposition coding based schemes in split-rate space and show the presence of polymatroidal structure. 

The achievability schemes of this work on two nested multicast messages can be seen as being related to the framework in \cite{Romero-Varanasi:isit2017}. However, the emphasis here is on carefully choosing relatively simple rate-splitting strategies tailored to the two nested multicast messages problem, employing partial interference decoding via non-unique, rather than unique, decoding, and on obtaining explicit two-dimensional descriptions for the achievable rate regions in terms of the rates of the two messages.
Such descriptions are then shown to be amenable to proving converse theorems that establish the capacity region for various classes of DM BCs for general $K$ and $L$. 




\subsection{An Overview of Results}
\label{overview}

In the first main result of this paper, 
we characterize three classes of DM BCs for which the natural extension of the Korner-Marton scheme with superposition coding 
is capacity achieving. The first class is one in which there is a private receiver that is less capable than each of the other private receivers and less noisy than each of the common receivers. The second class is one in which there is a common receiver that is more noisy than all other receivers. The technical contribution in both these cases is the proof of the converse. Since the Csiszar sum lemma has no generalization for more than two receivers, we use the information inequality \cite[Lemma 1]{nair2011capacity} to bypass this problem. This inequality was originally developed in the context of the 3-receiver DM BC with private messages \cite{nair2011capacity}, where it was used to obtain a tight outer bound for the class of 3-receiver less noisy DM BCs defined by $Y_3  \prec  Y_2  \prec  Y_1$. 
Finally, the third class for which the natural extension of the Korner-Marton scheme is shown to be capacity-optimal is one where there is a common receiver that is more noisy than each of the other common receivers and there is a private receiver that is less capable than each of the other private receivers. 

In the special cases of (a) a single private receiver with $\mathcal{S}=\{1 \}$ and (b) single common receiver so that $\mathcal{S}=\{1,  \cdots, K-1 \}$, the first and second classes of DM BCs mentioned above reduce to ones in which (a) $Y_j  \prec Y_1$ for each $j\in \{2,3,..,K\}$ and (b) $Y_K \prec  Y_j$ for each $j\in\{1,2,..,K{-}1\}$, thereby expanding \textemdash by replacing degradedness conditions with the less noisy conditions\textemdash the corresponding classes of DM BCs for which capacity was found in \cite{diggavi2006opportunistic}. In another direction, Case (b) is an extension of the corresponding result obtained in \cite[Proposition 11]{nair2009capacity} from $K=3$ to general $K$. The third class of DM BCs for which capacity is established can be seen as an extension of the capacity result due to Korner-Marton  \cite{kiirner1977general} for $K=2$ and $L=1$ to general $K$ and $L$.

In the second main result of his paper, drawing inspiration from the achievable scheme for the $K=3$ and $L=1$ case in \cite{nair2009capacity}, we propose an inner bound based on rate splitting, superposition coding, and indirect decoding for general $K$ and $L$. In particular, we split the private message $M_{1}$ into $K{-}L$ sub-messages, with one sub-message for every common receiver (except the last one) through which it indirectly decodes the common message and one sub-message to be decoded only by the private receivers. For this rate-splitting strategy of the private message, a closed-form description of polyhedral achievable rate region per coding distribution would require projecting away an indeterminate number of split-rate variables for general $K$. When $K=3$ this projection can be done by hand using Fourier-Motzkin Elimination (FME) \cite{schrijver1986theory} as in \cite{nair2009capacity}. For general $K$, we identify a certain structure in the region's inequalities and use that structure to eliminate the split rates. Consequently, we present a closed-form expression for our inner bound in terms of the original message rates. 

The third main contribution of this work is to consider the specialization of rate-splitting and superposition coding to one that splits the private message into only two sub-messages. One (private) sub-message is to be decoded only by the private receivers and the other (semi-public) sub-message by a subset of the common receivers. Those common receivers decode the common message by indirectly decoding the satellite codeword for the semi-public sub-message, which is superimposed on the cloud codewords for the common message. The rest of the common receivers decode the common message directly from the cloud centers. For this specialization, we obtain a converse result that uses both the Csiszar sum lemma and the information inequality to establish the capacity region for several sub-classes of three or more receiver DM BCs for which it was hitherto unknown. Moreover,  this result includes 
the capacity region for the class of three-receiver DM BCs with the restriction $Y_3 \prec Y_1 $ of  \cite{nair2009capacity}. 




The rest of this paper is organized as follows. In Section \ref{Sec_Def}, we formally state the problem set-up and establish notation. Section \ref{Sec_Main_results} is devoted for the three main results following which the paper is concluded in Section \ref{Sec_Conclusion}. Detailed proofs of the first two results are given in Appendices \ref{Appendix_Proof_Th1}, \ref{Appendix_Proof_general} and \ref{Appendix_proof_FME_structure}. 

\section{Preliminaries And Definitions}
\label{Sec_Def}

We consider a DM BC consisting of one transmitter $X\in \mathcal{X}$, $K$ receivers $Y_i\in \mathcal{Y}_i$, and the channel transition probability $W(y_1...y_K|x)$ where the conditional probability of $n$ channel outputs ($Y_{1j}...Y_{Kj}), \enspace j \in \{1,...,n\}$, conditioned on $n$ channel inputs ($X_{1}...X_n$) is given by 
\begin{equation}
p(y_1^n...y_K^n|x^n)= \prod_{j=1}^nW(y_{1j}...y_{Kj}|x_j)
\end{equation}

As previously stated, we study the problem of groupcasting two nested multicast messages $M_{0}$ and $ M_1$ over a $K$-user DM BC with $M_{0}$ to be decoded by all the receivers while the private message $M_{1}$ by the subset of $L$ private receivers indexed by $\mathcal{S}{=}\{1, 2 , \cdots  , L\}$. The set of indices of the common receivers that need to only decode $M_0$ is denoted as $\mathcal{C}{=}\{L{+}1, \cdots , K\}$.

The messages $M_{0}$ and $M_{1}$ have rates $R_{0}$ and $R_{1}$, respectively. A ($2^{nR_{0}}$, $2^{nR_{1}}$, $n$) code consists of an encoder that assigns a codeword $x^n(m_{0},m_{1})$ for each message pair $(m_{0},m_{1})$ $\in [1:2^{nR_{0}}]\times[1:2^{nR_{1}}]$, and  $K$ decoders with the $j^{th}$ private decoder mapping the received sequences $Y_{j,1}^n$ for each $j \in \mathcal{S}$ into an estimate ($\hat{m}_{0}^{(j)},\hat{m}_{1}^{(j)}$) and the $i^{th}$ common receiver mapping $Y_{i,1}^n$ into an estimate $\hat{m}_{0}^{(i)}$ for each $i \in \mathcal{C}.$ The probability of error $P_e^{(n)}$ is the probability that not all receivers decode their messages correctly. The rate pair ($R_{0}$,$R_{1}$) is said to be achievable if there exists a sequence of ($2^{nR_{0}}$,$2^{nR_{1}}$,$n$) codes with $P_e^{(n)}{\rightarrow} 0$ as $n {\rightarrow} \infty$. The closure of the union of achievable rates is the capacity region.


\section{Main Results}
\label{Sec_Main_results}

This section is organized along the lines of the overview of results in Section \ref{overview}. In Section \ref{sec-general-converse}, we provide three conditions for the optimality of the superposition coding/joint decoding for the general case of the $K$-user DM BC with $L$ private receivers in Theorem \ref{Th_Ext_KM_region}. 
The second main result is a  general inner bound 
based on splitting the private message into $K-L$ sub-messages and superposition coding, and is presented in Section \ref{sec-inner-bound}. In
Section \ref{Simple Rate-Splitting} a simplification of the general inner bound based on splitting the private message into just two sub-messages is presented and shown via a converse result that it is the capacity region of $K-L-1$ classes of DM BCs. 
In Section \ref{geq3rs}, we discuss the technical difficulty of proving a converse for a larger class of channels than for those found in Section \ref{Simple Rate-Splitting} using more than two splits of the private message.



Since the converse proofs in this section use the information inequality of \cite[Lemma 1]{nair2011capacity}, we state it here for easy reference.\\

\begin{lemma}
\label{Lemma_Information_Inequality}
Let $X {\markov} (Y,Z)$ be a DM BC without feedback and $Z \prec  Y$. Consider $M$ to be any random variable such that
$M {\markov}X^n {\markov} (Y^n,Z^n)$ forms a Markov chain. Then, we have \begin{align}
I(Y^{i-1};Z_{i}|M)\geq & I(Z^{i-1};Z_{i}|M) \enspace 1\leq i\leq n \nonumber \\
I(Y^{i-1};Y_{i}|M)\geq & I(Z^{i-1};Y_{i}|M) \enspace 1\leq i\leq n \nonumber
\end{align}
\end{lemma}
.\\



\subsection{Superposition Coding}
\label{sec-general-converse}

In the next theorem, we state the three classes of $K$-receiver DM BCs for which superposition coding alone is capacity-optimal for general two nested messages. 

\begin{theorem}\label{Th_Ext_KM_region}
For the $K$-receiver DM BC with two general nested multicast messages and $L$ private receivers the set of rate pairs ($R_0,R_1$) satisfying 
\begin{align}
R_0&\leq I(U;Y_c)  \; \forall c\in \mathcal{C} \nonumber \\
R_1&\leq I(X;Y_s|U) \; \forall s\in \mathcal{S} \nonumber \\
R_0{+}R_1&\leq I(X;Y_s) \; \forall s\in \mathcal{S}
\label{3rdineqthm1}
\end{align} for some $p(u,x)$ is the capacity region for the following cases:
\begin{itemize}
\item[(i)] $\exists $ a $r \in \mathcal{S}$ such that $Y_r \preceq Y_s$ $ \forall s\in \mathcal{S}\backslash \{r\}$ and $Y_c \prec  Y_r$ $ \forall c\in \mathcal{C}$.
\item[(ii)] $\exists $ a $ j \in \mathcal{C}$ such that $Y_j \prec  Y_c$ $ \forall c \in \mathcal{C}\backslash \{j\}$ and $Y_j \prec  Y_s$  $ \forall s \in \mathcal{S}$.
\item[(iii)] $\exists$ a $j \in \mathcal{C}$ such that $Y_j \prec Y_c \enspace \forall c  \in \mathcal{C}\backslash \{j\}$ and $ \exists r \in \mathcal{S}$ such that $Y_r \preceq Y_s$ $ \forall s  {\in} \mathcal{S}\backslash \{r\} $
\end{itemize}
\end{theorem}
\begin{proof}
The region in \eqref{3rdineqthm1} is the extension of the two-receiver Korner-Marton region in \eqref{Region_KM_2} to the $K$-receiver DM BC when all private users use joint decoding.

In Case (i), when $Y_r \preceq Y_s$ for $s \in \mathcal{S}\backslash \{r\} $, the last two sets of inequalities of \eqref{3rdineqthm1} become $R_1{\leq} I(X;Y_r|U)$ and  $R_0{+}R_1{\leq} I(X;Y_r)$. Moreover, since $Y_c \prec  Y_r$, the latter inequality is redundant. Hence, the region defined by \eqref{3rdineqthm1} becomes the set of rate pairs ($R_0,R_1$) satisfying 
\begin{align}
R_0&\leq I(U;Y_c)  \; \forall c\in \mathcal{C} \nonumber \\
R_1&\leq I(X;Y_r|U)  \label{2ndineqthm1i} 
\end{align}
for some $p(u,x)$. The converse proof for the above two inequalities is given in Appendix \ref{Appendix_Proof_Th1}. It is inspired from the converse proofs in \cite{nair2012three} where the information inequality plays a crucial role. To the reader familiar with that inequality it is clear that the optimal choice of the auxiliary random variable would be $U_i{=}M_0,Y_{r,1}^{i-1}$. It does not involve $Y_{c}^n$, and hence, the proof holds for each $c\in \mathcal{C}$.

In Case (ii), when $Y_j \prec  Y_c$ for all $c \in \mathcal{C} \backslash\{j\}$, the first set of inequalities in \eqref{3rdineqthm1} become $R_0\leq I(U;Y_j)$. Moreover, since $Y_j \prec  Y_s$ for all $s \in \mathcal{S}$, the last set of inequalities in \eqref{3rdineqthm1} becomes redundant.
Hence, the rate region in \eqref{3rdineqthm1} can be written as the set of rate pairs ($R_0,R_1$) satisfying 
\begin{align}
R_0&\leq I(U;Y_j)  \;  \nonumber \\
R_1&\leq I(X;Y_s|U) \; \forall s\in \mathcal{S} \label{2ndineqthm1ii}
\end{align}
for some $p(u,x)$. The converse for these two inequalities are also presented in Appendix \ref{Appendix_Proof_Th1}, where the optimal choice of the auxiliary random variable is $U_i{=}M_0,Y_{j,1}^{i-1}$, which is not a function of $Y_s^n$. Note that the information inequality is used in this converse proof as well.

Lastly, for the class of channels in Case (iii), the achievable region defined by \eqref{3rdineqthm1} becomes the set of rate pairs ($R_0,R_1$) satisfying\footnote{For Case (iii) joint decoding is needed at the private receivers whereas for Cases (i) and (ii) successive decoding suffices.}
\begin{align}
R_0&\leq I(U;Y_j)  \;  \nonumber \\
R_1&\leq I(X;Y_r|U) \;  \nonumber \\
R_0{+}R_1&\leq I(X;Y_r) \; \label{3rdineqthm1iii}
\end{align} for some $p(u,x)$, which is equivalent to the two-receiver Korner-Marton region in \eqref{Region_KM_2}. In particular, it is the capacity region for DM BC when only receivers $Y_j$ and $Y_r$ are present. 
Next, we note that the capacity region of the two-receiver DM BC with receivers $Y_j$ and $Y_r$ is an outer bound for the $K$-receiver DM BC under consideration since the addition of the other receivers cannot enlarge the capacity region. Hence, since the region \eqref{3rdineqthm1iii} is also achievable (with superposition coding), it is the capacity region for the $K$-receiver DM BC. Note that the converse proof for the two-receiver case in \cite{kiirner1977general} uses the Csiszar sum lemma, i.e., the optimal choice of the auxiliary random variable is $U_i{=}M_0,Y_{j,1}^{i-1}Y_{r,i+1}^n$. 

\end{proof}

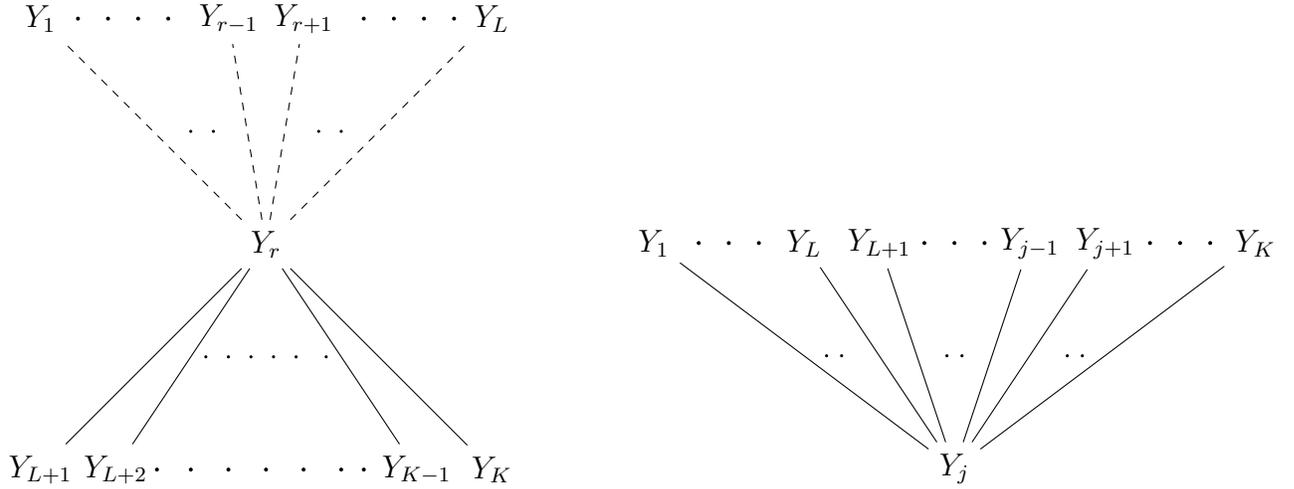
\begin{figure*}[t!]
    \centering
    \begin{subfigure}[t]{0.5\textwidth}
        \centering
\begin{tikzpicture}
 \node (y1) at (-3,6)  {$Y_1$};
  \node (yj_1) at (-0.5,6)  {$Y_{r-1}$};
  \node (yj+1) at (0.5,6)  {$Y_{r+1}$};
  \node (yL) at (3,6)  {$Y_{L}$};
 \node (yj) at (0,3)  {$Y_r$};
  \node (yL_1) at (-3,0)  {$Y_{L+1}$};
  \node (yL_2) at (-2,0)  {$Y_{L+2}$};
  \node (yK_1) at (2,0)  {$Y_{K-1}$};
  \node (yK) at (3,0)  {$Y_K$};
 
 \draw [-] (yj) -- (yL_1);
 \draw [-] (yj) -- (yL_2);
\draw [-] (yj) -- (yK);
 \draw [-] (yj) -- (yK_1);
 
 \draw [dashed] (yj) -- (y1);
 \draw [dashed] (yj) -- (yj_1);
\draw [dashed] (yj) -- (yj+1);
 \draw [dashed] (yj) -- (yL);

\draw[black,fill=black] (-1.45,0) circle (.1ex);
\draw[black,fill=black] (-1,0) circle (.1ex);
\draw[black,fill=black] (-0.5,0) circle (.1ex);
\draw[black,fill=black] (0,0) circle (.1ex);
\draw[black,fill=black] (1,0) circle (.1ex);
\draw[black,fill=black] (1.3,0) circle (.1ex);
\draw[black,fill=black] (0.5,0) circle (.1ex);

\draw[black,fill=black] (-0.8,1.5) circle (.07ex);
\draw[black,fill=black] (-0.5,1.5) circle (.07ex);
\draw[black,fill=black] (-0.2,1.5) circle (.07ex);
\draw[black,fill=black] (0.1,1.5) circle (.07ex);
\draw[black,fill=black] (.4,1.5) circle (.07ex);
\draw[black,fill=black] (0.8,1.5) circle (.07ex);


\draw[black,fill=black] (-2.5,6) circle (.1ex);
\draw[black,fill=black] (-2.1,6) circle (.1ex);
\draw[black,fill=black] (-1.7,6) circle (.1ex);
\draw[black,fill=black] (-1.3,6) circle (.1ex);
\draw[black,fill=black] (1.3,6) circle (.1ex);
\draw[black,fill=black] (1.7,6) circle (.1ex);
\draw[black,fill=black] (2.1,6) circle (.1ex);
\draw[black,fill=black] (2.5,6) circle (.1ex);

\draw[black,fill=black] (-1,4.5) circle (.07ex);
\draw[black,fill=black] (-0.7,4.5) circle (.07ex);
\draw[black,fill=black] (1,4.5) circle (.07ex);
\draw[black,fill=black] (0.7,4.5) circle (.07ex);

\end{tikzpicture}
        \caption{The class of channels of Case (i) of Theorem \ref{Th_Ext_KM_region}: $\exists $ a $r \in \mathcal{S}$ such that $Y_r \preceq Y_s$ $ \forall s\in \mathcal{S}\backslash \{r\}$ and $Y_c \prec  Y_r$ $ \forall c\in \mathcal{C}$.}
   \label{Fig:Classes_Th1_i}
   \end{subfigure}%
    ~ 
    \begin{subfigure}[t]{0.5\textwidth}
        \centering
        \begin{tikzpicture}
  \node (yj) at (0,1.5)  {$Y_j$};
  \node (y1) at (-4,4.5)  {$Y_1$};
  \node (yL) at (-2,4.5)  {$Y_L$};
  \node (yL+1) at (-1,4.5)  {$Y_{L+1}$};
  \node (yj-1) at (1,4.5)  {$Y_{j-1}$};
  \node (yj+1) at (2,4.5)  {$Y_{j+1}$};
  \node (yK_2) at (4,4.5)  {$Y_K$};

\draw [-] (yj) -- (y1);
\draw [-] (yj) -- (yL);
\draw [-] (yj) -- (yL+1);
\draw [-] (yj) -- (yj-1);
\draw [-] (yj) -- (yj+1);
\draw [-] (yj) -- (yK_2);

\draw[black,fill=black] (-3.4,4.5) circle (.1ex);
\draw[black,fill=black] (-3,4.5) circle (.1ex);
\draw[black,fill=black] (-2.6,4.5) circle (.1ex);

\draw[black,fill=black] (0.4,4.5) circle (.1ex);
\draw[black,fill=black] (0,4.5) circle (.1ex);
\draw[black,fill=black] (-0.4,4.5) circle (.1ex);

\draw[black,fill=black] (3.4,4.5) circle (.1ex);
\draw[black,fill=black] (3,4.5) circle (.1ex);
\draw[black,fill=black] (2.6,4.5) circle (.1ex);

\draw[black,fill=black] (-1.7,3) circle (.07ex);
\draw[black,fill=black] (-1.5,3) circle (.07ex);

\draw[black,fill=black] (-0.1,3) circle (.07ex);
\draw[black,fill=black] (0.1,3) circle (.07ex);

\draw[black,fill=black] (1.7,3) circle (.07ex);
\draw[black,fill=black] (1.5,3) circle (.07ex);

\end{tikzpicture}

        \caption{The class of channel of Case (ii) of Theorem \ref{Th_Ext_KM_region}: $\exists $ a $ j \in \mathcal{C}$ such that $Y_j \prec  Y_c$ $ \forall c \in \mathcal{C}\backslash \{j\}$ and $Y_j \prec  Y_s$  $ \forall s \in \mathcal{S}$}
   \label{Fig:Classes_Th1_ii}
   \end{subfigure}
\caption{The solid line between any two receivers indicates that the upper receiver is less noisy than the lower one while the dashed line between any two receivers indicates that the upper receiver is more capable than the lower one. If there is no line between any two receivers, there is no order assumed between them.} 
\label{Fig:Classes_Th1_i_ii}
\end{figure*}

\begin{figure*}[t!]
 \centering
\begin{tikzpicture}
 \node (y1) at (-3,6)  {$Y_1$};
  \node (yr_1) at (-0.5,6)  {$Y_{r-1}$};
  \node (yr+1) at (0.5,6)  {$Y_{r+1}$};
  \node (yL) at (3,6)  {$Y_{L}$};
 
 \node (yL+1) at (5,6)  {$Y_{L+1}$};
  \node (yj-1) at (7.5,6)  {$Y_{j-1}$};
  \node (yj+1) at (8.5,6)  {$Y_{j+1}$};
  \node (yK) at (11,6)  {$Y_{K}$};
\node (yj) at (8,3)  {$Y_{j}$};

 \node (yr) at (0,3)  {$Y_r$};
 
 \draw [dashed] (yr) -- (y1);
 \draw [dashed] (yr) -- (yr_1);
\draw [dashed] (yr) -- (yr+1);
 \draw [dashed] (yr) -- (yL);

\draw [-] (yj) -- (yL+1);
 \draw [-] (yj) -- (yj-1);
\draw [-] (yj) -- (yj+1);
\draw [-] (yj) -- (yK);

\draw[black,fill=black] (-2.5,6) circle (.1ex);
\draw[black,fill=black] (-2.1,6) circle (.1ex);
\draw[black,fill=black] (-1.7,6) circle (.1ex);
\draw[black,fill=black] (-1.3,6) circle (.1ex);
\draw[black,fill=black] (1.3,6) circle (.1ex);
\draw[black,fill=black] (1.7,6) circle (.1ex);
\draw[black,fill=black] (2.1,6) circle (.1ex);
\draw[black,fill=black] (2.5,6) circle (.1ex);

\draw[black,fill=black] (-1,4.5) circle (.07ex);
\draw[black,fill=black] (-0.7,4.5) circle (.07ex);
\draw[black,fill=black] (1,4.5) circle (.07ex);
\draw[black,fill=black] (0.7,4.5) circle (.07ex);

\draw[black,fill=black] (5.6,6) circle (.1ex);
\draw[black,fill=black] (5.9,6) circle (.1ex);
\draw[black,fill=black] (6.3,6) circle (.1ex);
\draw[black,fill=black] (6.7,6) circle (.1ex);
\draw[black,fill=black] (9.3,6) circle (.1ex);
\draw[black,fill=black] (9.7,6) circle (.1ex);
\draw[black,fill=black] (10.1,6) circle (.1ex);
\draw[black,fill=black] (10.5,6) circle (.1ex);

\draw[black,fill=black] (7,4.5) circle (.07ex);
\draw[black,fill=black] (7.3,4.5) circle (.07ex);
\draw[black,fill=black] (9,4.5) circle (.07ex);
\draw[black,fill=black] (8.7,4.5) circle (.07ex);

\end{tikzpicture}
        \caption{ The class of channels of Case (iii) of Theorem \ref{Th_Ext_KM_region}: $ \exists r \in \mathcal{S}$ such that $Y_r \preceq Y_s$ $ \forall s  {\in} \mathcal{S}\backslash \{r\} $ (the left side of the figure) and $\exists$ a $j \in \mathcal{C}$ such that $Y_j \prec Y_c \enspace \forall c  \in \mathcal{C}\backslash \{j\}$ (the right side of the figure).}
   \label{Fig:Classes_Th1iii}
   \end{figure*}
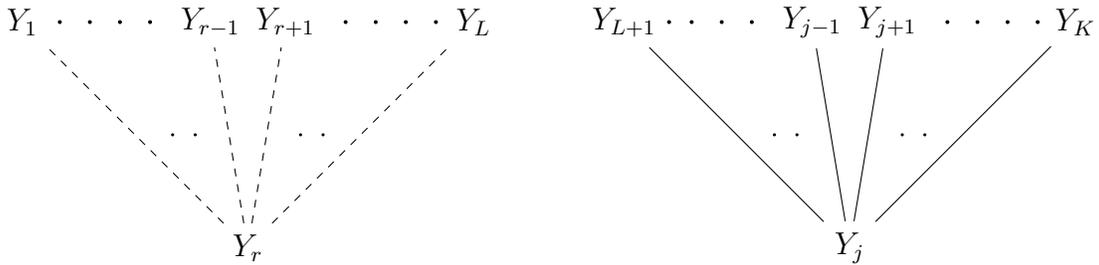

In Figs. \ref{Fig:Classes_Th1_i_ii} and \ref{Fig:Classes_Th1iii}, we illustrate the classes of channels considered in Theorem \ref{Th_Ext_KM_region} with a Hasse diagram each, that allows for a quick visualization of each class. In Fig. \ref{Fig:Classes_Th1_i}, the conditions of Case (i) are depicted: among the private receivers, there is one which is pair-wise less capable (shown with dashed lines) than all other private receivers and less noisy (shown with solid lines) than all the common receivers. In Fig. \ref{Fig:Classes_Th1_ii} the conditions of Case (ii) are shown: there is one common receiver that is pair-wise more noisy than all the other common receivers and all private receivers. Fig. \ref{Fig:Classes_Th1iii} shows the class of channels in Case (iii): there is one common receiver that is pair-wise more noisy than all the other common receivers and there is one private receiver which is pair-wise less capable than all other private receivers.

\begin{remark}
Although we are considering the K-receiver DM BC, the obtained capacity regions for the three cases of Theorem \ref{Th_Ext_KM_region} depend only on the channels of a subset of the receivers. For instance, in Case (i), the region depends on only $K-L+1$ receivers, while in Case (ii) it depends on $L+1$ receivers, and in Case (iii) it depends on only two receivers. In particular, the addition of any number of less noisy common receivers $Y_c$ (with $Y_j \prec Y_c$) in Cases (ii) and (iii) and/or any number of more capable private receivers $Y_s$ (with $Y_r \preceq Y_s$) in Cases (i) and (iii) does not change the capacity region from what it was before adding these receivers. Note that in Case (iii) of channels there is no order restriction across the common and private receivers, unlike in Cases (i) and (ii), i.e., restrictions are only needed among the common receivers and among the private receivers.
\end{remark}


\begin{remark}
Theorem \ref{Th_Ext_KM_region} includes several results known in the literature. These are listed below.
\begin{itemize}
\item Consider Case (iii) of Theorem \ref{Th_Ext_KM_region} for $K=2$ and $L=1$. In this case, there is no restriction on the channel and the capacity region is given as \eqref{3rdineqthm1iii} (with $r=1$ and $j=2$) which is the Korner-Marton capacity region for the two-receiver case.
\item  Consider general $K$ but $L=1$. This case was addressed in \cite{diggavi2006opportunistic} where optimality was shown for the class of DM BCs under the degraded order restrictions $X{\markov} Y_1 {\markov} Y_c, \forall c \in \mathcal{C}$. Case (i) of Theorem \ref{Th_Ext_KM_region} gives the capacity region for the larger class of channels, $Y_c\prec Y_1\; \forall c\in \mathcal{C}$. This class if depicted in Fig \ref{Fig:Classes_for_L_1_}.
\item Consider general $K$ but $L=K-1$. This case was addressed in \cite{diggavi2006opportunistic} where optimality was shown for the class of DM BCs under the degraded order restrictions $X{\markov} Y_s {\markov} Y_K, \forall s \in \mathcal{S}$. Case (ii) of Theorem \ref{Th_Ext_KM_region} expands this result to the class of channels $Y_K\prec Y_s, \forall s \in \mathcal{S}$, shown in Fig \ref{Fig:Classes_for_L_K_1}. 

\item Case (ii) of Theorem \ref{Th_Ext_KM_region} is a generalization to general $K$ and $L$ of the result in \cite[Proposition 11]{nair2009capacity} for $K=3$ and $L=2$ which obtained capacity for DM BCs defined by the two restrictions $ Y_3 \prec Y_1 $ and $Y_3 \prec Y_2$.
Note also that the converse proof of Case (ii) of Theorem \ref{Th_Ext_KM_region} mainly depends on the information inequality, while that in \cite[Proposition 11]{nair2009capacity} mainly depends on the Csiszar sum lemma.


\end{itemize}
\end{remark}

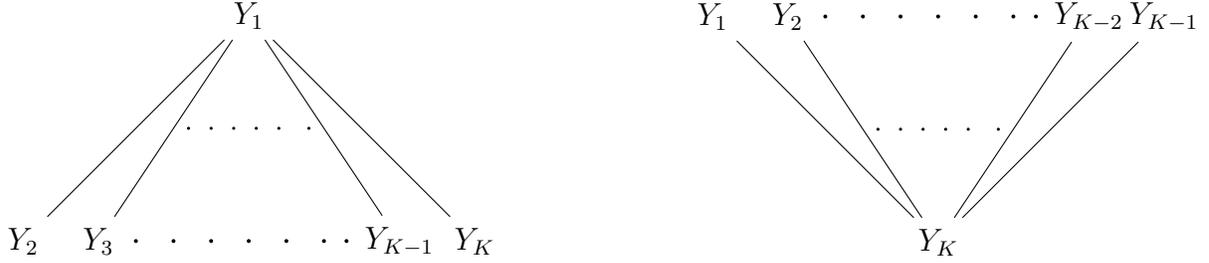
\begin{figure*}[t!]
    \centering
    \begin{subfigure}[t]{0.5\textwidth}
        \centering
\begin{tikzpicture}
 \node (y1) at (0,3)  {$Y_1$};
  \node (y2) at (-3,0)  {$Y_2$};
  \node (y3) at (-2,0)  {$Y_3$};
  \node (yK_1) at (2,0)  {$Y_{K-1}$};
  \node (yK) at (3,0)  {$Y_K$};
     \draw [-] (y1) -- (y2);
    \draw [-] (y1) -- (yK);
\draw [-] (y1) -- (y3);
    \draw [-] (y1) -- (yK_1);
\draw[black,fill=black] (-1.5,0) circle (.1ex);
\draw[black,fill=black] (-1,0) circle (.1ex);
\draw[black,fill=black] (-0.5,0) circle (.1ex);
\draw[black,fill=black] (0,0) circle (.1ex);
\draw[black,fill=black] (1,0) circle (.1ex);
\draw[black,fill=black] (1.3,0) circle (.1ex);
\draw[black,fill=black] (0.5,0) circle (.1ex);

\draw[black,fill=black] (-0.8,1.5) circle (.07ex);
\draw[black,fill=black] (-0.5,1.5) circle (.07ex);
\draw[black,fill=black] (-0.2,1.5) circle (.07ex);
\draw[black,fill=black] (0.1,1.5) circle (.07ex);
\draw[black,fill=black] (.4,1.5) circle (.07ex);
\draw[black,fill=black] (0.8,1.5) circle (.07ex);

\end{tikzpicture}
        \caption{The class of channel of Case (i) of Theorem \ref{Th_Ext_KM_region} with $L=1$ and, without loss of generality, letting $j=1$ ($Y_c\prec Y_1$ for all $c\in \{2,3,\cdots,K\}$)}
   \label{Fig:Classes_for_L_1_}
   \end{subfigure}%
    ~ 
    \begin{subfigure}[t]{0.5\textwidth}
        \centering
        \begin{tikzpicture}
 \node (yK) at (0,0)  {$Y_K$};
  \node (y1) at (-3,3)  {$Y_1$};
  \node (y2) at (-2,3)  {$Y_2$};
  \node (yK_2) at (2,3)  {$Y_{K-2}$};
  \node (yK_1) at (3,3)  {$Y_{K-1}$};

\draw [-] (yK) -- (y2);
\draw [-] (yK) -- (y1);
\draw [-] (yK) -- (yK_2);
\draw [-] (yK) -- (yK_1);

\draw[black,fill=black] (-1.5,3) circle (.1ex);
\draw[black,fill=black] (-1,3) circle (.1ex);
\draw[black,fill=black] (-0.5,3) circle (.1ex);
\draw[black,fill=black] (0,3) circle (.1ex);
\draw[black,fill=black] (1,3) circle (.1ex);
\draw[black,fill=black] (1.3,3) circle (.1ex);
\draw[black,fill=black] (0.5,3) circle (.1ex);

\draw[black,fill=black] (-0.8,1.5) circle (.07ex);
\draw[black,fill=black] (-0.5,1.5) circle (.07ex);
\draw[black,fill=black] (-0.2,1.5) circle (.07ex);
\draw[black,fill=black] (0.1,1.5) circle (.07ex);
\draw[black,fill=black] (.4,1.5) circle (.07ex);
\draw[black,fill=black] (0.8,1.5) circle (.07ex);

\end{tikzpicture}

        \caption{The class of channel of Case (ii) of Theorem \ref{Th_Ext_KM_region} with $L=K-1$ and letting $r=K$ ($Y_K\prec Y_s$ for all $s\in \{1,2,\cdots,K-1\}$)}
   \label{Fig:Classes_for_L_K_1}
   \end{subfigure}
\caption{The special cases of Theorem 1 for $L=1$ and $L=K-1$.} 
\label{Fig:Classes_for_L_1_and_L_K_1}
\end{figure*}
 
In \cite{nair2009capacity}, the capacity region was established for $K=3$ and $L=1$ for the class of channels $Y_2 \prec  Y_1$ 
(or $Y_3 \prec  Y_1$). 
However, Case (i) of Theorem \ref{Th_Ext_KM_region} gives the capacity for a strictly more restrictive class of DM BCs defined by $Y_2 \prec  Y_1$ {\em and} $Y_3 \prec  Y_1$. The reason for this is that the achievable scheme in \cite{nair2009capacity} is more complex in that it involves rate-splitting and indirect decoding, which brings us to the next section. In particular, we will generalize the scheme of \cite{nair2009capacity} to general $K$ and $L$ next.

\subsection{Rate-Splitting and Superposition Coding: A New Inner Bound}
\label{sec-inner-bound}
We expand our study of the capacity regions of classes of $K$-user DM BCs with two nested multicast messages for general $L$ by incorporating rate splitting before superposition coding.

\begin{theorem}
\label{TH_genral_case}
An inner bound of $K$-user DM BC for two nested multicast messages ($M_{0}$,$M_{1}$) with $M_0$ to be decoded by all receivers and $M_1$ to be decoded by the first $L$ private receivers is the set of rate pairs ($R_{0},R_{1}$) satisfying
\begin{align}
\label{Region_Genral_Case}
R_{0} \leq & I(U_{c};Y_c)  \enspace \forall c\in \mathcal{C} \nonumber \\
R_1\leq & I(X;Y_s|U_K)  \enspace \forall s\in \mathcal{S} \nonumber \\
R_{0}{+}R_{1} \leq & I(X;Y_s|U_{c}){+}I(U_{c};Y_c)  \enspace  \forall s\in \mathcal{S} \enspace \forall c\in \mathcal{C} \backslash \{K\} \nonumber \\
R_{0}{+}R_{1} \leq & I(X;Y_s) \enspace  \forall s\in \mathcal{S}
\end{align}
for some $(U_{K},U_{K{-}1},...,U_{L{+}1},X)$ such that $U_{K} {\markov}$ $U_{K{-}1} {\markov}$ $U_{K{-}2} ... {\markov}$ $U_{L{+}1}{\markov} X$ forms a Markov chain. This inner bound is the capacity region for the DM BCs of Cases (i)-(iii) of Theorem \ref{Th_Ext_KM_region}.
\end{theorem}
\begin{proof} We give an outline of the proof here and provide the detailed version in Appendix \ref{Appendix_Proof_general}. Split the private message $M_{1}$ into $K{-}L$ parts. There is a part for every common receiver except $Y_K$, with $M_{1i}$ for each $i{\in}\{L+1,L+2,..,K{-}1\}$ and one part, denoted $M_{11}$, to be decoded only by the private receivers. $M_{0}$ is represented by $U_{K}$, then $M_{1K{-}1}$ is superposed on $U_K$ to obtain $U_{K{-}1}$, and then $M_{1K{-}2}$ is superposed on $U_{K{-}1}$ to yield $U_{K{-}2}$, and so on. Finally, ($M_{0},M_{1}$) is represented by $X$. The receiver $Y_K$ finds $M_0$ by decoding $U_K$, the common receiver $Y_c$ finds $M_{0}$ by indirectly decoding $U_c$ for each $c \in \{L+1,..,K-1\}$, and the private receivers find ($M_{0},M_{1}$) by decoding $X$. 

From the proof in Appendix \ref{Appendix_Proof_general}, it is clear that the superposition coding only 
region of \eqref{3rdineqthm1} is contained in (\ref{Region_Genral_Case}) because we can just set all the split rates to zero in (\ref{Error_common_Rec}) and (\ref{Error_private_Rec}). Hence, for the classes of channels for which region \eqref{3rdineqthm1} is optimal as given by the three cases of Theorem \ref{Th_Ext_KM_region}, the rate region of Theorem \ref{TH_genral_case} reduces to region \eqref{3rdineqthm1} and is the capacity region. 
\end{proof}

\begin{remark}
Theorem \ref{TH_genral_case} represents one generalization of the achievability scheme proposed for $K=3$ and $L=1$ in \cite[Theorem 1]{nair2009capacity}. Other generalizations are possible. Most generally, one can split the private message into $2^{K-L}
$ parts as was done by the authors in \cite{Salman2018Ach}. However, a key feature of the scheme in Theorem \ref{TH_genral_case} is that the reliability conditions in split-rate space have a structure that allows us to analytically project away the split rates using the FME method to get a closed-form expression for the polygonal achievable rate region. This then sets the stage for proving converse results for new classes of DM BCs as in Theorem \ref{Th_Group_InDirect_Decoding} to follow. Whereas, projecting away exponentially many split rates for the scheme in \cite{Salman2018Ach} doesn't appear to lend itself to an explicit description for general $K$ and $L$. The ``brute-force" FME procedure, though tedious, can however be done when $K-L \leq 2$ \cite{Salman2018Ach}.  
\end{remark}

\begin{remark}
The authors in \cite[Theorem 1]{nair2009capacity} showed that indirect decoding is unnecessary in the $K{=}3,L{=}1$ case. This raises the question of whether the achievable rate region of Theorem \ref{TH_genral_case} remains unchanged if the common receivers employed unique decoding\footnote{This question is addressed by Proposition \ref{Prop_No_Need_IndirectDecoding} in Appendix \ref{app:nudnec?} in the more limited context of Corollary \ref{Corollary_OneIndirect} of Theorem \ref{TH_genral_case} to follow.}. Note however that the structure of the inequalities for non-unique decoding was used to project away the split rates to prove Theorem \ref{TH_genral_case}. Hence, even if non-unique decoding is unnecessary in Theorem \ref{TH_genral_case}, it is unclear if the (greater number of) inequalities for unique decoding would lend themselves to analytically eliminating split rates. Moreover, even if this were possible, the explicit description of the polygonal rate regions in that inner bound may have a description that could render proving the converse result of Theorem \ref{Th_Group_InDirect_Decoding} (to follow) intractable. 

\end{remark}

Next, we provide three examples of achievable rate regions resulting from Theorem \ref{TH_genral_case}. 


\begin{example}
\label{ex:K=3L=1}
Consider the smallest non-trivial case $K{=}3, L{=}1$ which was studied previously in \cite{nair2009capacity}. If we specialize the region (\ref{Region_Genral_Case}) to it, we get the region of rate pairs
\begin{align}
R_0{\leq}& \min \{I(U_3;Y_3),I(U_2,Y_2)\} \nonumber \\
R_1{\leq}& I(X;Y_1|U_3) \nonumber \\
R_0{+}R_1&{\leq} I(X;Y_1|U_2){+}I(U_2;Y_2)\nonumber \\
R_0{+}R_1&{\leq} I(X;Y_1)
\label{capacityK=3L=1}
\end{align}
for some $U_3 \markov U_2 \markov X$.
\end{example}
The above region, when further specialized to the class of DM BCs for which $Y_3 \prec Y_1$, is equivalent to that in \cite[Theorem 1]{nair2009capacity}. 

\begin{example}
\label{ex:K=4L=1}
Consider a case where $K{=}4 $ and $ L{=}1$. For this case, we get from (\ref{Region_Genral_Case}) that the set of rates for which
\begin{align}
R_0&{\leq} \min \{I(U_4;Y_4),I(U_3,Y_3),I(U_2,Y_2)\} \nonumber \\
R_1&{\leq} I(X;Y_1|U_4) \nonumber \\
R_0{+}R_1&{\leq} \min \{I(X;Y_1|U_2){+}I(U_2;Y_2), I(X;Y_1|U_3){+}I(U_3;Y_3)\} \nonumber \\
R_0{+}R_1&{\leq}  I(X;Y_1) \label{exk=4l=1}
\end{align}
for some $ U_4 \markov U_3 \markov U_2 \markov X$ is achievable.
\end{example}

\begin{example}
\label{ex:K=4L=2}
Consider next $K{=}4 $ and $ L{=}2$. For this case, we get from (\ref{Region_Genral_Case}) that the set of rates for which
\begin{align}\label{four-user-example}
R_0&{\leq} \min \{I(U_4;Y_4),I(U_3,Y_3)\} \nonumber \\
R_1&{\leq} \min \{I(X;Y_1|U_4),I(X;Y_2|U_4) \}\nonumber \\
R_0{+}R_1&{\leq} \min \{I(X;Y_1|U_3),I(X;Y_2|U_3)\}{+}I(U_3;Y_3)\nonumber \\
R_0{+}R_1&{\leq} \min \{ I(X;Y_1),I(X;Y_2)\} 
\end{align}
for some $ U_4 \markov U_3 \markov X$ is achievable.

\end{example}

The optimality of the regions in Examples \ref{ex:K=3L=1}-\ref{ex:K=4L=2} for classes of channels beyond those described by 
Theorem \ref{Th_Ext_KM_region} is addressed in Theorem \ref{Th_Group_InDirect_Decoding} in the next section. Example \ref{ex:K=4L=1} is further discussed in Section \ref{geq3rs}.

\subsection{Simple Rate-Splitting and A Converse}
\label{Simple Rate-Splitting}

Next, Theorem \ref{TH_genral_case} is specialized by splitting the private message into two parts instead of $K-L$ parts. Consequently, the resulting rate region is shown to be the capacity region for several classes of DM BCs using the information inequality {\em and} the Csiszar sum lemma in the proof of the converse. A discussion of the 
apparent insufficiency of these inequalities for proving a converse for more than two rate splits is given in Section \ref{geq3rs}.

With just two split rates, the most general scheme is one in which one group of common receivers decodes the common message indirectly while the rest of the common receivers decode it directly. 

\begin{corollary}
\label{Corollary_2_groups}
Fix $ l \in [L+1 : K] $. Partition the common  receivers into the two groups $\mathcal{C}_1 =  \{L+1,..,l\}$ and $\mathcal{C}_2 = \{l+1,..,K\}$. 
Then, the set of rate pairs ($R_0,R_1$) satisfying  
\begin{align}
\label{Region_prop_group_indirect_decoding}
R_0&\leq \min \{ I(U_{l} ; Y_{c_1}), I(U_K;Y_{c_2}\}, \enspace c_1\in \mathcal{C}_1 \enspace  c_2 \in \mathcal{C}_2 \nonumber \\
R_1&\leq I(X;Y_s|U_K), \enspace s\in \mathcal{S} \nonumber \\
R_0{+}R_1&\leq I(X;Y_s|U_{l}){+}I(U_{l} ; Y_{c_1}), \enspace c_1 \in \mathcal{C}_1 \enspace  s\in \mathcal{S} \nonumber \\
R_0{+}R_1&\leq I(X;Y_s),  \enspace  s\in \mathcal{S} 
\end{align} for some $U_K{\markov} U_{l}{\markov} X$ is achievable.
\end{corollary}
\begin{proof} We specialize the proof of Theorem \ref{TH_genral_case} in Appendix \ref{Appendix_Proof_general} by splitting the private message into just two parts so that $R_1{=}R_{11}{+}R_{1l}$ (i.e., set $R_{1,j}=0 \enspace \forall j\in \mathcal{C} \backslash \{l,K\}$\footnote{Note that $R_{1K}$ is always equal to zero \eqref{Error_common_Rec}.}). The region described in \eqref{Region_prop_group_indirect_decoding} hence follows by eliminating the single split rate variable using FME, a simpler task than the proof in Appendix \ref{Appendix_Proof_general}. 
Hence, when $l \in [L+1 : K-1] $, to achieve the region \eqref{Region_prop_group_indirect_decoding} the common receivers 
in $\mathcal{C}_1$ find the common message indirectly by decoding $U_{l}$ while those in $\mathcal{C}_2$ find $M_0$ directly by decoding $U_K$. Evidently, the variable $l$ controls the numbers of common receivers that decode $M_0$ directly and indirectly. When $l=K$, there is no rate-splitting. Here, all common receivers decode the common message directly from $U_K$ so that we just get the rate region of \eqref{3rdineqthm1} due to superposition coding alone. 
\end{proof}


\begin{theorem}
\label{Th_Group_InDirect_Decoding}
Let $ l \in [L+1 : K] $ and partition the common  receivers into $\mathcal{C}_1 =  \{L+1,..,l\}$ and $\mathcal{C}_2 = \{l+1,..,K\}$.
Consider the class of channels, shown in Fig \ref{Fig:Classes_Th3}, for which there exists a $j \in \mathcal{C}_1$ such that  $Y_j \prec  Y_{c_1}$ for all $c_1 \in \mathcal{C}_1 \backslash \{j\} $ and a $r \in \mathcal{S}$ such that $Y_r \preceq Y_s$ $\forall s \in \mathcal{S}\backslash \{r\} $, and $Y_{c_2} \prec  Y_r$  $ \forall c_2 \in \mathcal{C}_2 $. For this class of channels,
 the union of rate pairs ($R_0,R_1$) satisfying 
\begin{align}
\label{Region_prop_group_indirect_decoding_spec}
R_0&\leq \min \{ I(U_{l} ; Y_{j}), I(U_K;Y_{c_2}\}  \enspace  c_2 \in \mathcal{C}_2 \nonumber \\
R_1&\leq I(X;Y_r|U_K)  \nonumber \\
R_0{+}R_1&\leq I(X;Y_r|U_{l}){+}I(U_{l} ; Y_{j}) \nonumber \\ 
R_0{+}R_1&\leq I(X;Y_r) 
\end{align} 
for some $U_K{\markov} U_{l}{\markov} X$ is the capacity region.
\end{theorem}

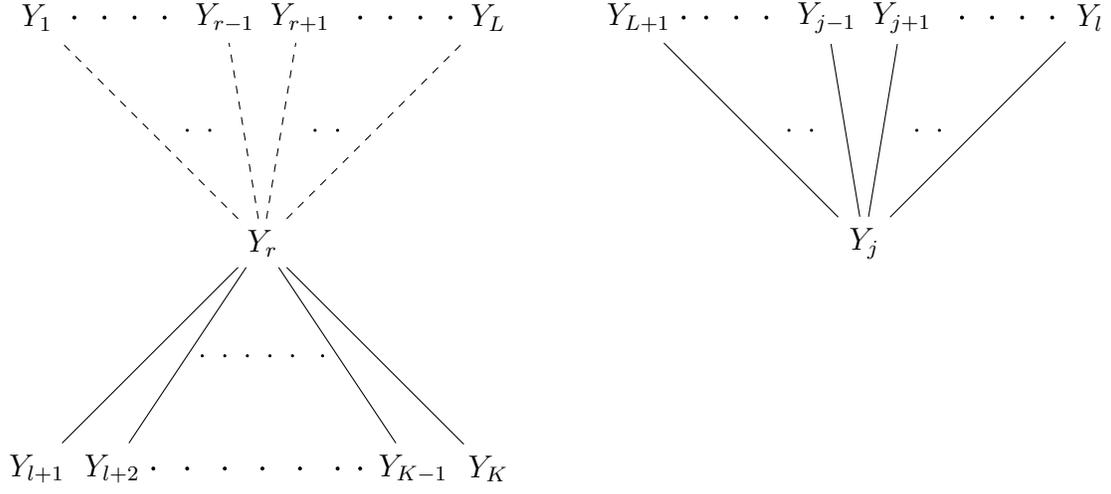
\begin{figure*}[t!]
 \centering
\begin{tikzpicture}
 \node (y1) at (-3,6)  {$Y_1$};
  \node (yr_1) at (-0.5,6)  {$Y_{r-1}$};
  \node (yr+1) at (0.5,6)  {$Y_{r+1}$};
  \node (yL) at (3,6)  {$Y_{L}$};
 
 \node (yL+1) at (5,6)  {$Y_{L+1}$};
  \node (yj-1) at (7.5,6)  {$Y_{j-1}$};
  \node (yj+1) at (8.5,6)  {$Y_{j+1}$};
  \node (yl) at (11,6)  {$Y_{l}$};
\node (yj) at (8,3)  {$Y_{j}$};

 \node (yr) at (0,3)  {$Y_r$};
  \node (yl_1) at (-3,0)  {$Y_{l+1}$};
  \node (yl_2) at (-2,0)  {$Y_{l+2}$};
  \node (yK_1) at (2,0)  {$Y_{K-1}$};
  \node (yK) at (3,0)  {$Y_K$};
 
 \draw [dashed] (yr) -- (y1);
 \draw [dashed] (yr) -- (yr_1);
\draw [dashed] (yr) -- (yr+1);
 \draw [dashed] (yr) -- (yL);

\draw [-] (yj) -- (yL+1);
 \draw [-] (yj) -- (yj-1);
\draw [-] (yj) -- (yj+1);
\draw [-] (yj) -- (yl);

 \draw [-] (yr) -- (yl_1);
 \draw [-] (yr) -- (yl_2);
\draw [-] (yr) -- (yK);
\draw [-] (yr) -- (yK_1);

\draw[black,fill=black] (-1.45,0) circle (.1ex);
\draw[black,fill=black] (-1,0) circle (.1ex);
\draw[black,fill=black] (-0.5,0) circle (.1ex);
\draw[black,fill=black] (0,0) circle (.1ex);
\draw[black,fill=black] (1,0) circle (.1ex);
\draw[black,fill=black] (1.3,0) circle (.1ex);
\draw[black,fill=black] (0.5,0) circle (.1ex);

\draw[black,fill=black] (-0.8,1.5) circle (.07ex);
\draw[black,fill=black] (-0.5,1.5) circle (.07ex);
\draw[black,fill=black] (-0.2,1.5) circle (.07ex);
\draw[black,fill=black] (0.1,1.5) circle (.07ex);
\draw[black,fill=black] (.4,1.5) circle (.07ex);
\draw[black,fill=black] (0.8,1.5) circle (.07ex);

\draw[black,fill=black] (-2.5,6) circle (.1ex);
\draw[black,fill=black] (-2.1,6) circle (.1ex);
\draw[black,fill=black] (-1.7,6) circle (.1ex);
\draw[black,fill=black] (-1.3,6) circle (.1ex);
\draw[black,fill=black] (1.3,6) circle (.1ex);
\draw[black,fill=black] (1.7,6) circle (.1ex);
\draw[black,fill=black] (2.1,6) circle (.1ex);
\draw[black,fill=black] (2.5,6) circle (.1ex);

\draw[black,fill=black] (-1,4.5) circle (.07ex);
\draw[black,fill=black] (-0.7,4.5) circle (.07ex);
\draw[black,fill=black] (1,4.5) circle (.07ex);
\draw[black,fill=black] (0.7,4.5) circle (.07ex);

\draw[black,fill=black] (5.6,6) circle (.1ex);
\draw[black,fill=black] (5.9,6) circle (.1ex);
\draw[black,fill=black] (6.3,6) circle (.1ex);
\draw[black,fill=black] (6.7,6) circle (.1ex);
\draw[black,fill=black] (9.3,6) circle (.1ex);
\draw[black,fill=black] (9.7,6) circle (.1ex);
\draw[black,fill=black] (10.1,6) circle (.1ex);
\draw[black,fill=black] (10.5,6) circle (.1ex);

\draw[black,fill=black] (7,4.5) circle (.07ex);
\draw[black,fill=black] (7.3,4.5) circle (.07ex);
\draw[black,fill=black] (9,4.5) circle (.07ex);
\draw[black,fill=black] (8.7,4.5) circle (.07ex);

\end{tikzpicture}
        \caption{ The class of channel in the statement of Theorem \ref{Th_Group_InDirect_Decoding}: there exists a $j \in \mathcal{C}_1$ such that  $Y_j \prec  Y_{c_1}$ for all $c_1 \in \mathcal{C}_1 \backslash \{j\} $ (upper right part of figure), an $r \in \mathcal{S}$ such that $Y_r \preceq Y_s$ $\forall s \in \mathcal{S}\backslash \{r\} $ (upper left part of figure), and $Y_{c_2} \prec  Y_r$  $ \forall c_2 \in \mathcal{C}_2 $ (lower left part of figure). 
        }
   \label{Fig:Classes_Th3}
   \end{figure*}

\begin{proof} The proof of achievability can be deduced from Corollary \ref{Corollary_2_groups} by removing the redundant inequalities in \eqref{Region_prop_group_indirect_decoding} for the class of channels under consideration in this theorem.

The proof of converse for the first two inequalities of \eqref{Region_prop_group_indirect_decoding_spec} follow the same technique (which depends on the information inequality) as that in Theorem \ref{Th_Ext_KM_region}. The last inequality is straightforward. The only non-trivial inequality is the third one. We show that the optimal identification of the auxiliary random variables is $U_{Ki}{=}M_0Y_{r,1}^{i-1}$ and $U_{li}=M_0Y_{r,1}^{i-1}Y_{j,i+1}^n$. To prove the third inequality, observe that for any sequence of $(2^{nR_0}, 2^{nR_1}, n)$ codes with $\lim_{n \rightarrow \infty} P_e^{(n)} = 0 $, we have
\begin{align}
n (R_{0}{+}R_1)&=H(M_0)+H(M_1) \nonumber \\
& \leq I(M_0;Y_{j}^n){+} I(M_1;M_0,Y_r^n){+} n \epsilon_n 
 \label{chainrule7}\\
& {=} \sum_{i=1}^n  I(M_0;Y_{j,i}|Y_{j,i+1}^n)  \nonumber \\ & 
+  \sum_{i=1}^n I(M_{1};Y_{r,i}|M_0,Y_{r,1}^{i-1}) {+}n \epsilon_n
\label{fano4} \\
&{\leq} \sum_{i=1}^n I(M_0,Y_{j,i+1}^n;Y_{j,i}) \nonumber\\
&+  \sum_{i=1}^n I(M_{1};Y_{r,i}|M_0,Y_{r,1}^{i-1}) {+}n \epsilon_n \nonumber \\
&{\leq} \sum_{i=1}^n I(M_0,Y_{r,1}^{i-1},Y_{j,i+1}^n;Y_{j,i}) \nonumber \\
&- \sum_{i=1}^n I(Y_{r,1}^{i-1};Y_{j,i}|M_0,Y_{j,i+1}^n) \nonumber \\
&{+} \sum_{i=1}^n I(M_{1};Y_{r,i}|M_0,Y_{r,1}^{i-1},Y_{j,i+1}^n) \nonumber \\
& + \sum_{i=1}^n I(Y_{j,i+1}^n;Y_{r,i}|M_0,Y_{r,1}^{i-1}) 
{+}n \epsilon_n \label{Due_to_chain_rule}\\
&{=}\sum_{i=1}^n I(M_{1};Y_{r,i}|M_0,Y_{r,1}^{i-1},Y_{j,i+1}^n) \nonumber \\
&+ \sum_{i=1}^n I(M_0,Y_{r,1}^{i-1},Y_{j,i+1}^n;Y_{j,i})+n \epsilon_n \label{Due_to_scizar_sum_lemma} \\
&{\leq } \sum_{i=1}^n I(X_i;Y_{r,i}|U_{l,i}){+}I(U_{l,i};Y_{j,i}){+}n \epsilon_n \nonumber 
\end{align}
where (\ref{chainrule7}) follows from Fano's inequality, the chain rule and non-negativity of conditional mutual information, and (\ref{fano4}) from chain rule for mutual information and the independence between $M_1$ and $M_0$. Inequality (\ref{Due_to_chain_rule}) also follows from the chain rule, and (\ref{Due_to_scizar_sum_lemma}) is due to the Csiszar sum lemma. Now letting $Q\in \{1,2,...,n\}$ to be a uniformly distributed random variable independent of all other random variables and defining $U_K=(U_{K,Q},Q), U_{l}=(U_{l,Q},Q)$, and $X=X_Q$, the converse proof can be completed in the standard way. Note that we have $U_K{\markov}U_{l}{\markov}X$ since $U_{l} {=} (U_K,Y_{j,i+1}^n)$.


\end{proof}


\begin{remark}
For each $l \in [L+1:K-1]$, Theorem \ref{Th_Group_InDirect_Decoding} provides a distinct achievability region and establishes its capacity optimality for an associated distinct class of channels. Consequently, the general inner bound of Theorem \ref{TH_genral_case} is tight for the $K-L-1$ classes of channels characterized by each $l \in [L+1:K-1]$. When $l=K$, Theorem \ref{Th_Group_InDirect_Decoding} recovers Case (iii) of Theorem \ref{Th_Ext_KM_region}.
\end{remark}

The following examples illustrate the effect of varying $l$ on the capacity region and the class of channels for which capacity is achieved. 

\begin{example}
\label{Ex:K=4L=1-manyls}
Consider $K=4$ and $L=1$ and each of the three choices of $l\in\{2,3,4\}$. 

Let $l=K=4$ first. Theorem \ref{Th_Group_InDirect_Decoding} should reduce to Case (iii) of Theorem \ref{Th_Ext_KM_region}. Indeed, we get the result that the set of rate pairs ($R_0,R_1$) satisfying
\begin{align}
R_0&\leq I(U_4;Y_4)\nonumber\\
R_1 &\leq I(X;Y_1|U_4)\nonumber \\
R_0+R_1&\leq I(X;Y_1) \label{exthm3k=4l=1}
\end{align}
for some $p(u_4,x)$ is the capacity region when $Y_2\succ Y_4$ and $Y_3\succ Y_4$ (without loss of generality, we let $j=4$ in Theorem \ref{Th_Group_InDirect_Decoding}). Assuming that receiver $Y_4$ is more noisy than the other two common receivers essentially boils the problem down to the two-receiver problem in which only receivers $Y_1$ and $ Y_4$ are present. Indeed, the region in \eqref{exthm3k=4l=1} is the Korner-Marton region for the two-receiver DM BC (with receivers $(Y_1, Y_4)$) with no restriction between their channels.

Consider next the case of $l=3$. Theorem \ref{Th_Group_InDirect_Decoding} asserts that the union of rate pairs ($R_0,R_1$) satisfying
\begin{align*}
R_0&\leq \min \{I(U_3;Y_3),I(U_4;Y_4)\}\nonumber\\
R_1 &\leq I(X;Y_1|U_4)\nonumber \\
R_0+R_1&\leq I(X;Y_1|U_3)+I(U_3;Y_3)\nonumber 
\end{align*}
for some $U_4\markov U_3 \markov X$ is the capacity region when $Y_1\succ Y_4$ and $Y_2\succ Y_3$, (without loss of generality, we take $j=3$ in Theorem \ref{Th_Group_InDirect_Decoding}). The above region is hence the same as that in \cite[Proposition 7]{nair2009capacity} for the three-receiver problem with receivers ($Y_1,Y_3,Y_4$) with the single restriction $Y_1\succ Y_4$. The rationale behind this is that when $l=3$, receivers $Y_2$ and $Y_3$ decode $M_0$ indirectly from $U_2$. But since $Y_2\succ Y_3$, only the weaker receiver $Y_3$ matters since $Y_2$ can decode $M_0$ if $Y_3$ can, essentially reducing the problem to the three-receiver case with receivers $Y_1,Y_3$ and $Y_4$.

When $l=2$, we get from Theorem \ref{Th_Group_InDirect_Decoding} that the capacity region for the class of channels under the restrictions $Y_1\succ Y_3,Y_1\succ Y_4$ is the rate pairs ($R_0,R_1$) satisfying
\begin{align}
R_0&\leq \min \{I(U_2;Y_2),I(U_4;Y_3),I(U_4;Y_4)\}\nonumber\\
R_1 &\leq I(X;Y_1|U_4)\nonumber \\
R_0+R_1&\leq I(X;Y_1|U_2)+I(U_2;Y_2)\nonumber 
\end{align}
for some $U_4\markov U_2\markov X$. Note that the last inequality in \eqref{Region_prop_group_indirect_decoding_spec} is redundant. Hence, the above result can be seen as a generalization of \cite[Proposition 7]{nair2009capacity} from $K=3$ to $K=4$. 
 
\end{example}

\begin{example}
\label{Ex:K=6_L=1_different_l}
Consider next the case of $K=6$ and $L=1$. Here, consider $l\in\{2,3,4,5\}$. 
The classes of channels for these $l$ for which Theorem \ref{Th_Group_InDirect_Decoding} gives the capacity region are depicted in Fig. \ref{Fig:Classes_for_K=6_l=2_3_4_5}. For $l=2$, Fig. \ref{Fig:Classes_K=6_l=2} represents the conditions $Y_1 \succ Y_i$ for $i\in\{3,4,5,6\}$. For $l=3$ and $j=3$ in Theorem \ref{Th_Group_InDirect_Decoding}, Fig. \ref{Fig:Classes_K=6_l=3} represents the conditions $Y_1 \succ Y_i$ for $i\in\{4,5,6\}$ and $Y_2\succ Y_3$. For $l=4$ and $j=4$, Fig. \ref{Fig:Classes_K=6_l=4} represents the conditions $Y_1 \succ Y_i$ for $i\in\{5,6\}$ and $Y_{k}\succ Y_4$ for $k\in\{2,3\}$. Finally, for $l=5$ and $j=5$, Fig. \ref{Fig:Classes_K=6_l=5} represents the conditions $Y_1 \succ Y_6$ and $Y_{k}\succ Y_5$ for $k\in\{2,3,4\}$. 
\end{example}

The case $l=L+1$ of Theorem \ref{Th_Group_InDirect_Decoding} deserves explicit mention, which we state as a corollary.

\begin{figure*}[t!]
    \centering
    \begin{subfigure}[t]{0.2\textwidth}
        \centering
\begin{tikzpicture}
 \node (y1) at (0,3)  {$Y_1$};
  \node (y2) at (1,3)  {$Y_2$};
  \node (y3) at (-1.5,0)  {$Y_3$};
  \node (y4) at (-0.5,0)  {$Y_4$};
  \node (y5) at (0.5,0)  {$Y_5$};
  \node (y6) at (1.5,0)  {$Y_6$};
    \draw [-] (y1) -- (y3);
\draw [-] (y1) -- (y4);
    \draw [-] (y1) -- (y5);
    \draw [-] (y1) -- (y6);
\end{tikzpicture}
        \caption{The class of channels for $l=2$}
   \label{Fig:Classes_K=6_l=2}
   \end{subfigure}%
    ~ 
    \begin{subfigure}[t]{0.2\textwidth}
        \centering
        \begin{tikzpicture}
  \node (y1) at (0,3)  {$Y_1$};
  \node (y2) at (2,3)  {$Y_2$};
  \node (y3) at (2,0)  {$Y_3$};
  \node (y4) at (-1,0)  {$Y_4$};
  \node (y5) at (0,0)  {$Y_5$};
  \node (y6) at (1,0)  {$Y_6$};
    \draw [-] (y1) -- (y4);
    \draw [-] (y1) -- (y5);
    \draw [-] (y1) -- (y6);
    \draw [-] (y2) -- (y3);
    \end{tikzpicture}
        \caption{The class of channels for $l=3$}
   \label{Fig:Classes_K=6_l=3}
   \end{subfigure}
    \begin{subfigure}[t]{0.2\textwidth}
        \centering
\begin{tikzpicture}
 \node (y1) at (0,3)  {$Y_1$};
  \node (y2) at (1.5,3)  {$Y_2$};
  \node (y3) at (2.5,3)  {$Y_3$};
  \node (y4) at (2,0)  {$Y_4$};
  \node (y5) at (-0.5,0)  {$Y_5$};
  \node (y6) at (0.5,0)  {$Y_6$};
    \draw [-] (y1) -- (y5);
    \draw [-] (y1) -- (y6);
    \draw [-] (y2) -- (y4);
    \draw [-] (y3) -- (y4);
\end{tikzpicture}
        \caption{The class of channels for $l=4$}
   \label{Fig:Classes_K=6_l=4}
   \end{subfigure}%
    ~     \begin{subfigure}[t]{0.2\textwidth}
        \centering
\begin{tikzpicture}
 \node (y1) at (0,3)  {$Y_1$};
  \node (y2) at (1,3)  {$Y_2$};
  \node (y3) at (2,3)  {$Y_3$};
  \node (y4) at (3,3)  {$Y_4$};
  \node (y5) at (2,0)  {$Y_5$};
  \node (y6) at (0,0)  {$Y_6$};
    \draw [-] (y1) -- (y6);
\draw [-] (y2) -- (y5);
    \draw [-] (y3) -- (y5);
    \draw [-] (y4) -- (y5);
\end{tikzpicture}
        \caption{The class of channels of $l=5$}
   \label{Fig:Classes_K=6_l=5}
   \end{subfigure}%
\caption{The classes of channels for which Theorem \ref{Th_Group_InDirect_Decoding} gives capacity corresponding to different choices of $l$ when $K=6$ and $L=1$. See Example \ref{Ex:K=6_L=1_different_l}.} 
\label{Fig:Classes_for_K=6_l=2_3_4_5}
\end{figure*}
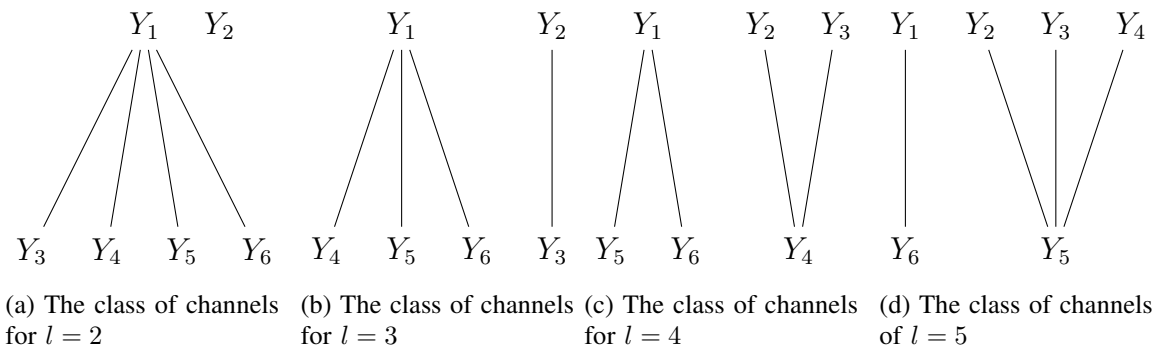

\begin{figure*}[t!]
 \centering
\begin{tikzpicture}
 \node (y1) at (-3,6)  {$Y_1$};
  \node (yr_1) at (-0.5,6)  {$Y_{r-1}$};
  \node (yr+1) at (0.5,6)  {$Y_{r+1}$};
  \node (yL) at (3,6)  {$Y_{L}$};
 
\node (yj) at (4,3)  {$Y_{L+1}$};
 
 \node (yr) at (0,3)  {$Y_r$};
  \node (yl_1) at (-3,0)  {$Y_{L+2}$};
  \node (yl_2) at (-2,0)  {$Y_{L+3}$};
\node (yK_1) at (2,0)  {$Y_{K-1}$};
  \node (yK) at (3,0)  {$Y_K$};

 \draw [dashed] (yr) -- (y1);
 \draw [dashed] (yr) -- (yr_1);
\draw [dashed] (yr) -- (yr+1);
 \draw [dashed] (yr) -- (yL);

 \draw [-] (yr) -- (yl_1);
 \draw [-] (yr) -- (yl_2);
\draw [-] (yr) -- (yK);
\draw [-] (yr) -- (yK_1);

\draw[black,fill=black] (-1.4,0) circle (.1ex);
\draw[black,fill=black] (-1.1,0) circle (.1ex);
\draw[black,fill=black] (-0.5,0) circle (.1ex);
\draw[black,fill=black] (0,0) circle (.1ex);
\draw[black,fill=black] (0.5,0) circle (.1ex);
\draw[black,fill=black] (1.4,0) circle (.1ex);
\draw[black,fill=black] (1.1,0) circle (.1ex);

\draw[black,fill=black] (-0.8,1.5) circle (.07ex);
\draw[black,fill=black] (-0.5,1.5) circle (.07ex);
\draw[black,fill=black] (-0.2,1.5) circle (.07ex);
\draw[black,fill=black] (0.1,1.5) circle (.07ex);
\draw[black,fill=black] (.4,1.5) circle (.07ex);
\draw[black,fill=black] (0.8,1.5) circle (.07ex);

\draw[black,fill=black] (-2.5,6) circle (.1ex);
\draw[black,fill=black] (-2.1,6) circle (.1ex);
\draw[black,fill=black] (-1.7,6) circle (.1ex);
\draw[black,fill=black] (-1.3,6) circle (.1ex);
\draw[black,fill=black] (1.3,6) circle (.1ex);
\draw[black,fill=black] (1.7,6) circle (.1ex);
\draw[black,fill=black] (2.1,6) circle (.1ex);
\draw[black,fill=black] (2.5,6) circle (.1ex);

\draw[black,fill=black] (-1,4.5) circle (.07ex);
\draw[black,fill=black] (-0.7,4.5) circle (.07ex);
\draw[black,fill=black] (1,4.5) circle (.07ex);
\draw[black,fill=black] (0.7,4.5) circle (.07ex);

\end{tikzpicture}
        \caption{ The class of channel in the statement of Corollary \ref{Corollary_OneIndirect}: $ \exists r \in \mathcal{S}$ 
such that $Y_r \preceq Y_s$ $ \forall s  {\in} \mathcal{S}\backslash \{r\} $ and $Y_c \prec  Y_r$ $ \forall c{\in} \mathcal{C} \backslash \{L+1 \} $. This figure is the same as Fig. \ref{Fig:Classes_Th3} if we let $l=L+1$, and hence $j=L+1$. } 
   \label{Fig:Classes_Corollary_l_L+1}
   \end{figure*}
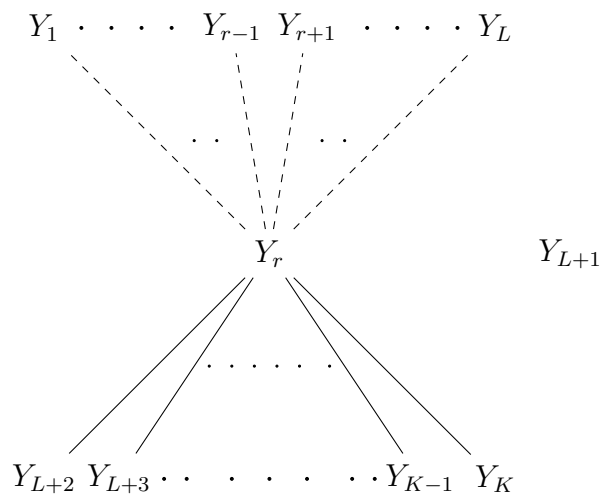

\begin{corollary}
\label{Corollary_OneIndirect}
For the class of DM BCs with two nested multicast messages 
defined by the restrictions that $ \exists r \in \mathcal{S}$ 
such that $Y_r \preceq Y_s$ $ \forall s  {\in} \mathcal{S}\backslash \{r\} $ and also $Y_c \prec  Y_r$ $ \forall c{\in} \mathcal{C} \backslash \{L+1 \} $, the union of rate pairs ($R_0,R_1$) satisfying 
\begin{align}
&R_0\leq \min \{I(U_K;Y_c),I(U_{L{+}1} ; Y_{L{+}1})\} \enspace \; \forall c{\in} \mathcal{C} \backslash \{L+1 \} \nonumber \\
&R_1\leq I(X;Y_r|U_K)  \nonumber \\
&R_0{+}R_1\leq I(X;Y_r|U_{L{+}1}){+}I(U_{L{+}1} ; Y_{L{+}1}) \label{Region_All_decode_Directly} 
\end{align} for some $U_K{\markov} U_{L{+}1}{\markov} X$ is the capacity region.   
\end{corollary}

\begin{proof} The proof follows directly from Theorem \ref{Th_Group_InDirect_Decoding} by setting $l=L+1$. Note that the last inequality in \eqref{Region_prop_group_indirect_decoding_spec} becomes redundant since the sum of the first two inequalities in \eqref{Region_All_decode_Directly} gives a more restrictive condition on the sum-rate.



\end{proof}

The class of channels for which Corollary \ref{Corollary_OneIndirect} gives capacity is depicted in Fig. \ref{Fig:Classes_Corollary_l_L+1}. Comparing it with Fig. \ref{Fig:Classes_Th1_i}, it is evident that it represents a strictly larger class of channels than that of Case (i) of Theorem \ref{Th_Ext_KM_region} since the restriction $Y_{L+1} \prec  Y_r$ is not needed. 

\begin{remark}
Examples \ref{ex:K=3L=1}, \ref{ex:K=4L=2} and the $l=2$ cases of Examples \ref{Ex:K=4L=1-manyls} and \ref{Ex:K=6_L=1_different_l} are  evidently examples of Corollary \ref{Corollary_OneIndirect}. In particular, according to Corollary \ref{Corollary_OneIndirect} the achievable rate region of Example \ref{ex:K=3L=1} (with $K=3$ and $L=1$) when $ Y_3 \prec Y_1$ (the last inequality in \eqref{capacityK=3L=1} becomes redundant under this condition) is the capacity region under that condition. This result was first obtained in \cite[Proposition 7]{nair2009capacity}. Corollary \ref{Corollary_OneIndirect} can thus be seen as one extension of \cite[Proposition 7]{nair2009capacity}. Note also that unlike the converse proof of \cite[Proposition 7]{nair2009capacity} which uses the Csiszar sum lemma, the converse proof of Corollary \ref{Corollary_OneIndirect} uses the Csiszar sum lemma {\em and} the information inequality, which together allows us to prove the more general result of Corollary \ref{Corollary_OneIndirect}.
\end{remark}

\begin{remark}
The authors in \cite[Theorem 1]{nair2009capacity} showed that indirect decoding is unnecessary in the $K{=}3,L{=}1$ case of Example \ref{ex:K=3L=1}. This raises the question of whether or not the achievable rate region of Corollary \ref{Corollary_OneIndirect} remains unchanged if the common receiver $Y_{L+1}$ employed unique decoding. This question is answered in the affirmative by Proposition \ref{Prop_No_Need_IndirectDecoding} in Appendix \ref{app:nudnec?}. Note that in \cite{chong2008han}, it was shown that the region obtained by non-unique decoding turned out to be equivalent to that of unique decoding in \cite{han1981new} in the context of the two-user interference channel. The result of Proposition \ref{Prop_No_Need_IndirectDecoding} amplifies somewhat the main message of \cite{bidokhti2014non}, which details other instances where indirect (or non-unique) decoding does not lead to a strict enlargement of the rate region over that with unique decoding. 
\end{remark}




We noted in the $l=3$ case of Example \ref{Ex:K=4L=1-manyls} in which $K=4$ and $L=1$ that under the restrictions of Theorem \ref{Th_Group_InDirect_Decoding} that the capacity region reduces to that of a three-receiver DM BC with appropriately chosen receivers. The next corollary shows this aspect of Theorem \ref{Th_Group_InDirect_Decoding} for general $K$ and $L$.

\begin{corollary}
\label{Corollary_l_Equals_K-1}
For the class of DM BCs with two nested multicast messages defined by the restrictions that $\exists j \in \{L+1,\cdots, K-1\}$ such that $Y_j \prec Y_c \enspace \forall c  \in \{L+1,\cdots, K-1\}\backslash \{j\}$ and $ \exists r \in \mathcal{S}$ such that $Y_r \preceq Y_s$ $ \forall s  {\in} \mathcal{S}\backslash \{r\} $ and $Y_K\prec Y_r$, the union of rate pairs ($R_0,R_1$) satisfying 
\begin{align}
&R_0\leq  \min \{I(U_{K-1} ; Y_{j}), I(U_K;Y_K)\} \nonumber \\
&R_1\leq I(X;Y_r|U_K)  \nonumber \\
&R_0{+}R_1\leq I(X;Y_r|U_{K-1})+I(U_{K-1};Y_j) \label{Region_All_decode_Directly2} 
\end{align} for some $U_K \markov U_{K-1} \markov X$ is the capacity region.   
\end{corollary}
\begin{proof} The proof follows from Theorem \ref{Th_Group_InDirect_Decoding} by setting $l=K-1$. A direct proof of this result is to consider the embedded three-receiver DM BC with the critical receivers $Y_r, Y_j $ and $Y_K$ (for which the conditions of the corollary on the channels implies a single condition $Y_K\prec Y_r$). Under this restriction the region of \eqref{Region_All_decode_Directly2} is the capacity region of this three-receiver DM BC. Since adding receivers can't enlarge the capacity region and since the region \eqref{Region_All_decode_Directly2} is also achievable in the $K$-receiver case, it must be the capacity region for the the $K$-receiver problem. 
\end{proof}

\begin{remark}
Note that Corollary \ref{Corollary_l_Equals_K-1} gives a region that is equivalent to that in \cite[Proposition 7]{nair2009capacity} for the three critical receivers ($Y_K,Y_j,Y_r$). The rationale behind this is as follows: when $l= K-1$, only receiver $Y_K$ decodes $M_0$ directly from $U_K$. The rest of common receivers $Y_c$ where $c  \in \{L+1,\cdots, K-1\}$ decode  $M_0$ indirectly from $U_{K-1}$. Now, the assumed restriction that receiver $Y_j$ is more noisy than each of these other common receivers means that if receiver $Y_j$ can decode the common message successfully, so can all those other common receivers. A similar argument can also applied to the private receivers since there is one private receiver that is less capable than all other private receivers.  
\end{remark}

Theorem \ref{Th_Group_InDirect_Decoding} can thus be seen as unifying and generalizing (for general $K$ and $L$) the important capacity results of Korner and Marton \cite{korner1977images} for $K=2$ and $L=1$ and of Nair and El Gamal \cite{nair2009capacity} for $K=3$ and $L=1$. As stated previously, the converse proof of Theorem \ref{Th_Group_InDirect_Decoding} uses both the Csiszar sum inequality and the information inequality. Can these inequalities be used to further expand the classes of channels for which capacity can be found by using the greater generality of Theorem \ref{TH_genral_case} which involves more than two rate splits for the private message?

\subsection{More Than Two Rate Splits: A Discussion}
\label{geq3rs}

In this section, we consider the $K=4$ and $L=1$ example and the rate region achievable according to Theorem \ref{TH_genral_case} given in \eqref{exk=4l=1}. That rate region involves splitting the private message into three sub-messages. We highlight the technical difficulty of proving a converse result for this region next.

Given that the capacity region for $K=3,L=1$ is only known when $Y_1\succ Y_3$ in \cite{nair2009capacity}, we expect that at least one less noisy condition is required to establish the converse proof for the region in \eqref{exk=4l=1} for $K=4,L=1$. Thus, let us assume that $Y_1\succ Y_4$. The last inequality in \eqref{exk=4l=1} is redundant in this case. By using the information inequality as in Case (i) of Theorem \ref{Th_Ext_KM_region}, we can prove the converse for the first part of the first inequality (i.e., $ R_0 {\leq} \min \{I(U_4;Y_4) $) as well as the second inequality in \eqref{exk=4l=1} by letting $U_{4i}=M_0,Y_{1,1}^{i-1}$. 
Next, since the Csiszar sum lemma does not require any restrictions on the channel, we can use it to prove the converse for the rest of the first inequality and the third inequality in \eqref{exk=4l=1} where 
the ``right" choice of the auxiliary random variables would be $U_{2i}=M_0Y_{1,1}^{i-1}Y_{2,i+1}^n$ and $U_{3i}=M_0Y_{1,1}^{i-1}Y_{3,i+1}^n$. This choice does not however match the Markov structure of the auxiliary random variables since we do not assume that $Y_3$ is a degraded version of $Y_2$. Moreover, there is no useful generalization of Csiszar sum lemma for $K>2$ that could be used to choose $U_{2i}=M_0Y_{1,1}^{i-1}Y_{2,i+1}^nY_{3,i+1}^n$, for instance, that satisfies the Markov structure of the auxiliary random variables. Due to this limitation, we have to add one more less noisy condition ($Y_1\succ Y_3$) that then allows us to use the information inequality instead. Hence, having the conditions $Y_1\succ Y_4$ and $Y_1\succ Y_3$ in hand, we can prove the converse for $K=4,L=1$ using a combination of the information inequality and Csiszar sum lemma. More precisely, the Csiszar sum lemma is used to choose $U_{2i}=M_0Y_{1,1}^{i-1}Y_{2,i+1}^n$ and the information inequality is used to choose $U_{3i}=M_0Y_{1,1}^{i-1}$ (given $Y_3 \prec Y_1$) and $U_{4i}=M_0Y_{1,1}^{i-1}$ (given $Y_4 \prec Y_1$). Note hence that $U_{3i}=U_{4i}$. This means that when we assume $Y_1\succ Y_4$ and $Y_1\succ Y_3$ we only need two auxiliary random variables instead of three, so that having only two split rates instead of three suffice to prove the converse which gets us back to Theorem \ref{Th_Group_InDirect_Decoding}. We finally note that the region in \eqref{exk=4l=1} for $K=4,L=1$ may not be optimal for channels with just a single restriction $Y_1\succ Y_4$, especially given the limited form of rate-splitting considered in Theorem \ref{TH_genral_case}. A more general rate splitting strategy as in \cite{Salman2018Ach} may have to be considered to find capacity for a class of channels less restrictive than the ones that satisfy $Y_1\succ Y_4$ and $Y_1\succ Y_3$. Such problems are of future interest.


\section{Conclusion}
\label{Sec_Conclusion}

$K$-user DM BCs with two nested multicast messages were studied in which one message is to be sent to all receivers and another to a subset of $L$ private receivers. The capacity region  was established for several classes of channels. We showed that a natural extension of Korner-Marton region for $K$-user DM BC, which is achievable by superposition coding with successive or joint decoding at the private receivers, is capacity optimal for a larger class of channels than previously found for $K \geq 3$ in the two cases in which there is either one private receiver or one common receiver. For the general case of $L$ private receivers, we obtain new classes of DM BCs for which the natural extension of Korner-Marton region with superposition coding is capacity optimal. Moreover, we obtain a new inner bound in explicit form that uses rate splitting, superposition coding and indirect decoding. The general form of rate-splitting we consider splits the private message into as many parts as there are common receivers. Particular specializations of it in which the private message is split into just two sub-messages and some subset of common receivers employ indirect decoding is shown to be optimal for several new classes of DM BCs. 



\begin{appendices}
\section{Proof of the converse for Theorem \ref{Th_Ext_KM_region}}
   \label{Appendix_Proof_Th1}

For Case (i) and Case (ii), the converse proof is based on the information inequality and is given next. In particular, we show that for every sequence of $(2^{nR_0}, 2^{nR_1}, n)$ codes with $\lim_{n \rightarrow \infty} P_e^{(n)} = 0 $ the inequalities in \eqref{2ndineqthm1i} or the inequalities in \eqref{2ndineqthm1ii} hold for some $p(u,x)$ for which $U \markov X \markov (Y_1,Y_2, \cdots , Y_K)$.

Consider case (i), i.e., the class of channels for which $\exists $ a $r \in \mathcal{S}$ such that $Y_r \preceq Y_s$ $ \forall s\in \mathcal{S}\backslash \{r\}$ and $Y_c \prec  Y_r$ $ \forall c\in \mathcal{C}$.
For the first inequality in \eqref{2ndineqthm1i} in Case (i), for each $ c\in \mathcal{C} $, we have
\begin{align}
nR_0&= H(M_0) \nonumber \\
& \leq I(M_0;Y_c^n) + n \epsilon_n  \label{fano1}\\
& = \sum_{i=1}^n I(M_{0};Y_{c,i}|Y_{c,1}^{i-1}) + n \epsilon_n \label{chainrule1} \\
& \leq \sum_{i=1}^n I(M_{0},Y_{c,1}^{i-1};Y_{c,i}) + n \epsilon_n \label{chainrule2}  \\ 
&\leq \sum_{i=1}^n I(M_{0},Y_{r,1}^{i-1};Y_{c,i}) + n \epsilon_n  \label{Y_1_LessNoisy_Lemma}\\ 
&=  \sum_{i=1}^n I(U_i;Y_{c,i}) + n \epsilon_n \label{auxiliary}
\end{align}
where (\ref{fano1}) follows from Fano's inequality (with  $\lim_{n \rightarrow \infty} \epsilon_{n} = 0 $), \eqref{chainrule1} and \eqref{chainrule2} from chain rule and non-negativity of mutual information, \eqref{Y_1_LessNoisy_Lemma} follows from an information inequality of Lemma \ref{Lemma_Information_Inequality} because $Y_c \prec Y_r$, and \eqref{auxiliary} by defining $U_i{=}M_{0},Y_{r,1}^{i-1}$ which satisfies $U_i \markov X_i \markov (Y_{r,i},Y_{c,i})$.

The rest of the proof proceeds along standard lines. Define a time-sharing uniform random variable $Q$ over $[1{:}n]$ that is independent of all other involved random variables. Identify $U{=}(U_Q,Q)$ and $Y_{c}{=}Y_{cQ}$ and take the limit as $n \rightarrow \infty$, so that $\epsilon_n \rightarrow 0$, to get the first inequality in \eqref{2ndineqthm1i}.

To prove the second inequality in \eqref{2ndineqthm1i}, we have 
\begin{align}
nR_1&= H(M_1) \nonumber \\
& \leq I(M_1;M_0, Y_r^n)+ n \epsilon_n  \label{fano-again} \\
& = \sum_{i=1}^n I(M_{1};Y_{r,i}|M_{0} Y_{r,1}^{i-1}) + n \epsilon_n  \label{indm0m1}  \\
&\leq \sum_{i=1}^n I(X_i;Y_{r,i}|U_i) + n \epsilon_n \label{useui} \\
& = n I(X_Q;Y_{r,Q}|U_Q, Q) + n \epsilon_n \label{atlast}
\end{align}
where 
\eqref{fano-again} follows from Fano's inequality, \eqref{indm0m1} from the independence of $M_0$ and $M_1$ and the chain rule for mutual information, \eqref{useui} from the definition of $U_i$, and \eqref{atlast} by introducing the time-sharing random variable. Identify $ X=X_Q,Y_{r}=Y_{r,Q}$ so that $U\markov X\markov (Y_r,Y_c)$ and take the limit as $n \rightarrow \infty$, to obtain \eqref{2ndineqthm1i}.

Consider Case (ii), i.e., the class of channels for which $\exists $ a $ j \in \mathcal{C}$ such that $Y_j \prec  Y_c$ $ \forall c \in \mathcal{C}\backslash \{j\}$ and $Y_j \prec  Y_s$  $ \forall s \in \mathcal{S}$. For the first inequality in \eqref{2ndineqthm1ii}, follow the steps leading to \eqref{chainrule2} to obtain
\begin{align}
nR_0&=H(M_0) \nonumber \\
      & \leq \sum_{i=1}^n I(M_{0},Y_{j,1}^{i-1};Y_{j,i}) + n \epsilon_n  \label{ChainRule_R0_2msg}  \\
      & = \sum_{i=1}^n I(U_i;Y_{j,i}) + n \epsilon_n  \label{auxiliary2}
\end{align}
where 
(\ref{auxiliary2}) follows by defining $U_{i}{=}M_{0},Y_{j,1}^{i-1}$, which satisfies $U_i\markov X_i \markov (Y_{s,i},Y_{j,i})$. 

To obtain the second inequality in \eqref{2ndineqthm1ii}, note that for any $ s\in \mathcal{S} $, following the steps leading 
to \eqref{indm0m1}, we have the first inequality below:
\begin{align}
nR_1 & \leq \sum_{i=1}^n I(M_{1};Y_{s,i}|M_{0} Y_{s,1}^{i-1}) + n \epsilon_n  \label{chainrule5}  \\
     & {=} \sum_{i=1}^n I(M_1,Y_{s,1}^{i-1};Y_{s,i}|M_{0})  {-}I(Y_{s,1}^{i-1};Y_{s,i}|M_{0}) {+}  n \epsilon_n \label{chainrule6} \\
     &{\leq}  \sum_{i=1}^n I(X_i;Y_{s,i}|M_{0} )  {-}I(Y_{s,1}^{i-1};Y_{s,i}|M_{0}) + n \epsilon_n \label{Markov_chain_Xi} \\
     & {\leq} \sum_{i=1}^n I(X_i;Y_{s,i}|M_{0} ){-}I(Y_{j,1}^{i-1};Y_{s,i}|M_{0}) + n \epsilon_n \label{Lemma_R_0_R_1} \\
     & {\leq}  \sum_{i=1}^n I(X_i;Y_{s,i}|M_{0}, Y_{j,1}^{i-1}) + n \epsilon_n \label{Markov_chain_Xi2}\\     
     &{=}  \sum_{i=1}^n I(X_i;Y_{s,i}|U_{i}) + n \epsilon_n \label{auxiliary3}
\end{align}
where (\ref{chainrule6}) follows from the chain rule for mutual information, \eqref{Markov_chain_Xi} and \eqref{Markov_chain_Xi2} from data processing inequality and that $(M_0,M_1,Y_{s,1}^{i-1},Y_{j,1}^{i-1}) {\markov} X_i{\markov} Y_{s,i} $. The inequality (\ref{Lemma_R_0_R_1}) follows from an information inequality of Lemma \ref{Lemma_Information_Inequality}, since $Y_j \prec Y_s$, and (\ref{auxiliary3}) from the definition of $U_i$. To complete the proof, we adopt the steps from \eqref{useui} to \eqref{atlast}, i.e., define a time-sharing uniform random variable $Q$ over $[1:n]$ that is independent of all other random variables, and identify $U=(U_Q,Q), X=X_Q, Y_s=Y_{s,Q} $, and $Y_j=Y_{j,Q}$ (so that $U \markov X \markov Y_s$) to obtain \eqref{2ndineqthm1ii}.

\section{Proof of Theorem 
\ref{TH_genral_case}}
   \label{Appendix_Proof_general}
The following lemmas are used later in the proof of Theorem \ref{TH_genral_case}. They may also be of independent interest.

\begin{lemma}
\label{Lemma_FME_R1k_R1l}
Consider the following set of inequalities for  $R_0$, $R_{1}$, and any integers $ 1\leq k\leq K-1$ and $K\geq 2$
\begin{align}
 R_{0}{+} \sum_{j{=}c}^{K}R_{1j}&\leq  I(U_{c};Y_{c}) \enspace \enspace c\in\{k,...,K\} \label{Eq_structure0_1}\\
 R_{1} {-} \sum_{j{=}c}^{K}R_{1j}&\leq  I(U;Y|U_{c})\enspace \enspace c\in\{k,...,K\}\label{Eq_structure0_2}\\
 -R_1+\sum_{j=k}^K R_{1j}&\leq 0 \label{Eq_structure0_3}\\
 -R_{1c}&\leq 0 \enspace \enspace c\in\{k,...,K\} \label{Eq_structure0_4}
\end{align} 
where $U_K\markov U_{K-1}\markov \cdots \markov U_{k} \markov U\markov (Y_K,Y_{K-1}...,Y_{k},Y)$ forms a Markov chain. Let $ k\leq l \leq K-1 $. By projecting away the $l-k+1$ ``split rate" variables ($R_{1k},R_{1k+1},...,R_{1l}$), we get the following polytope
\begin{align}
&R_0+R_1  \leq I(U;Y|U_c)+I(U_c;Y_c) \enspace c\in \{k,k+1,..,l\} \label{Eq_structure5_1}\\
& R_0+\sum_{j=l{+}1}^K R_{1j}\leq I(U_c;Y_c) \enspace c\in \{k,k+1,..,l\}\label{Eq_structure5_2}\\
&R_{0}{+} \sum_{j{=}c}^{K}R_{1j}\leq  I(U_{c};Y_{c}) \enspace \enspace c\in\{l+1,l+2.,..,K\} \label{Eq_structure5_3}\\
&R_{1} {-} \sum_{j{=}c}^{K}R_{1j}\leq  I(U;Y|U_{c})\enspace \enspace c\in\{l+1,l+2,...,K\}\label{Eq_structure5_4}\\
&-R_1+\sum_{j=l{+}1}^K R_{1j}\leq 0 \label{Eq_structure5_5}\\
&-R_{1c}\leq 0 \enspace \enspace c\in\{l{+}1,...,K\} \label{Eq_structure5_6}
\end{align}
\begin{proof} The proof uses mathematical induction and the FME method to sequentially eliminate $R_{1j}$, $j\in \{k,k+1,..,l\}$ and is given in Appendix \ref{Appendix_proof_FME_structure}. 
\end{proof} 
\end{lemma}

\begin{remark}
In the above lemma, we start with a polytope that consists of $K-k+1$ split rate variables described in \eqref{Eq_structure0_1}-\eqref{Eq_structure0_4}. By projecting out the first $l-k+1$ of them we get the polytope in smaller dimension described in \eqref{Eq_structure5_1}-\eqref{Eq_structure5_6}. Note that the set of inequalities in \eqref{Eq_structure5_3}-\eqref{Eq_structure5_6} looks similar to that in \eqref{Eq_structure0_1}-\eqref{Eq_structure0_4} with the only difference being that the index $c$ belonging to $\{l+1,...,K\}$ instead of $\{k,....,K\}$. On the other hand, the inequalities in \eqref{Eq_structure5_1} and \eqref{Eq_structure5_2} emerge due to the projection.
\end{remark}

\begin{lemma}
\label{Lemma_FME_R1_R0}
Consider the polytope of admissible rate pair $(R_0, R_{1})$ described by \eqref{Eq_structure0_1}-\eqref{Eq_structure0_4} in Lemma \ref{Lemma_FME_R1k_R1l} but allowing $ k \leq l \leq K$ under the Markov chain relation stated therein.
Its projection onto the $(R_0,R_1)$ plane is the polygon described by the following inequalities
\begin{align}
&R_{0}{+} R_1 {\leq}I(U;Y|U_{c}) {+} I(U_{c};Y_{c})\nonumber\\
&R_{0}{\leq}  I(U_{c};Y_{c}) \enspace c\in \{k,..,K\}
\label{Region_After_projection}
\end{align}
\begin{proof} The proof is divided into two steps. 

From Lemma \ref{Lemma_FME_R1k_R1l}, by setting $l=K-1$, we get 
\begin{align}
&R_0{+}R_1  {\leq} I(U;Y|U_c){+}I(U_c;Y_c) \enspace c{\in} \{k,k{+}1,..,K{-}1\} \label{Eq_structure7_1}\\
& R_0+ R_{1K}\leq I(U_c;Y_c) \enspace c\in \{k,k+1,..,K-1\}\label{Eq_structure7_2}\\
&R_{0}{+} R_{1K}\leq  I(U_{K};Y_{K})  \label{Eq_structure7_3}\\
&R_{1} {-}R_{1K}\leq  I(U;Y|U_{K}) \label{Eq_structure7_4}\\
&-R_1+R_{1K}\leq 0 \label{Eq_structure7_5}\\
&-R_{1K}\leq 0  \label{Eq_structure7_6}
\end{align}
Next, using FME, we project out the last split rate $R_{1K}$.  

Note that the first set of inequalities \eqref{Eq_structure7_1} does not contain $R_{1K}$, so that it will still the same after projecting out $R_{1K}$. Using FME, from \eqref{Eq_structure7_2} and \eqref{Eq_structure7_4}, we get
\begin{equation}
\label{project_away_R1K_1}
R_0{+}R_1 {\leq} I(U;Y|U_K){+}I(U_c;Y_c)\enspace c{\in} \{k,k{+}1,..,K{-}1\}
\end{equation}
Given \eqref{Eq_structure7_1}, these inequalities are redundant since $I(U;Y|U_{K}) \geq I(U;Y|U_{c}) \enspace \forall c\in \{k,..,K-1\}$, a consequence of $U_{K}\markov U_{K-1}\markov \cdots \markov U_{k} \markov U\markov (Y_K,Y_{K-1},...,Y_k,Y)$.

Again, using FME, projecting out $R_{1K}$ from \eqref{Eq_structure7_2} and \eqref{Eq_structure7_6}, we get 
\begin{equation}
\label{project_away_R1K_2}
R_0\leq I(U_c;Y_c) \enspace c\in \{k,k+1,..,K-1\}.
\end{equation}
Finally, by projecting $R_{1K}$ from \eqref{Eq_structure7_3} and \eqref{Eq_structure7_4}, we have
\begin{equation}
\label{project_away_R1K_3}
R_0+R_1\leq I(U;Y|U_K)+I(U_K;Y_K).
\end{equation}
and from \eqref{Eq_structure7_3} and \eqref{Eq_structure7_6} to get
\begin{equation}
\label{project_away_R1K_4}
R_0 \leq I(U_K;Y_K).
\end{equation}
Now, \eqref{Eq_structure7_1}, \eqref{project_away_R1K_2}, \eqref{project_away_R1K_3} and \eqref{project_away_R1K_4} together constitute the region in \eqref{Region_After_projection}. 
\end{proof} 
\end{lemma}

We now return to the proof of Theorem \ref{TH_genral_case}. The achievability scheme depends on superposition coding with rate splitting and indirect decoding. In particular, we split the private message $M_{1}$ into $K{-}L$ parts, with one part $M_{11}$ with rate $R_{11} $ to be decoded only by the private receivers and a part for every common receiver except $Y_K$, denoted $M_{1c}$, with rate $R_{1c}$ ($c {\in} \mathcal{C}\backslash \{K\}$). Fix a distribution $p(u_{K})\prod_{i{=}L{+}1}^{K{-}1}p(u_{i}|u_{i{+}1})p(x|u_{L{+}1})$. Generate $2^{nR_{0}}$ codewords $u_{K}^n(m_{0})$ according to $\prod_{i=1}^n p(u_{Ki})$. For each $u_{0}^n(m_{0})$, generate $2^{nR_{1K{-}1}}$ codewords  $u_{K{-}1}^n(m_{0}, m_{1K{-}1})$ according to $\prod_{i=1}^n p(u_{K{-}1i}|u_{Ki})$. Then, for every $u_c^n(m_{0},m_{1  K{-}1},...,m_{1  c})$ where $c {\in} \mathcal{C} \backslash \{K,L{+}1\} $, generate $2^{nR_{1 c{-}1}}$ codewords $u_{c{-}1}^n$ $(m_{0}, m_{1K{-}1},...,m_{1  c{-}1})$ according to $\prod_{i=1}^n p(u_{c{-}1i}|u_{ci})$. Finally, for every $u_{L{+}1}^n(m_{0}$ $, m_{1  K{-}1}$ $,...,m_{1  L{+}1})$, generate $2^{nR_{11}}$ codewords  $x^n$ $(m_{0}, m_{1})$ according to $\prod_{i=1}^n p(x_i|u_{L{+}1i})$. To send ($M_{0},M_{1}$), the sender expresses $M_{1}$ as ($M_{1 K{-}1},..,M_{1 L{+}1}$, $M_{1 1})$ and sends $x^n$ $(m_{0}$ $, m_{1 K{-}1}$ $,...,m_{1 L{+}1}$ $,m_{1 1}$). 

Receiver $Y_K$ finds $M_{0}$ if there is a unique index such that $u_K(m_{0})$ and $y_K^n$ are jointly typical. This happens with probability of correct decoding approaching unity as long as 
\begin{equation}
R_{0} {\leq} I(U_K;Y_K)
\label{Error_Simple_1}
\end{equation}The rest of the common receivers $Y_c$ ($c\in \{L{+}1,...,K{-}1\}$) find $M_{0}$ by indirectly decoding $U_c$. It is not difficult to show that this decoding succeeds with high probability as long as
\begin{equation}
 R_{0}{+} \sum_{j{=}c}^{K{-}1}R_{1j}{\leq}  I(U_{c};Y_{c}) \enspace c{\in} \{L{+}1,...,K{-}1\}
 \label{Error_common_Rec}
\end{equation} Finally, all private receivers find $M_0,M_1$ by decoding $X$. It can be shown that his happens successfully with high probability as long as the following inequalities hold:  
\begin{align}
\label{Error_simple_2}R_1\leq &I(X;Y_s|U_K)  \enspace s{\in} \mathcal{S} \\
\label{Error_private_Rec}  R_{1} {-} \sum_{j{=}c}^{K{-}1}R_{1j}\leq & I(X;Y_s|U_{c})  \enspace c{\in} \{L{+}1,...,K{-}1\},s{\in} \mathcal{S} \\ 
\label{Error_simple_3}R_{0} {+} R_{1} \leq & I(X;Y_s) \enspace s{\in} \mathcal{S}
\end{align} At this point, we need to eliminate the split rates from the sets of inequalities in (\ref{Error_common_Rec}) and (\ref{Error_private_Rec})
and the inequalities 
\begin{align}
&-R_1+\sum_{j=L+1}^{K-1} R_{1j}\leq 0 \label{Eq_sequence10_3}\\
&-R_{1c}\leq 0 \enspace \enspace c\in\{L+1,...,K-1\} \label{Eq_sequence10_4}
\end{align}
that guarantee that all split rates are positive to obtain the achievable region of Theorem \ref{TH_genral_case}.

The polytope described by (\ref{Error_common_Rec}), (\ref{Error_private_Rec}), \eqref{Eq_sequence10_3} and \eqref{Eq_sequence10_4} can be seen as an intersection of $L$  constituent polytopes, each corresponding to $s \in \mathcal{S}$. In order to obtain the projection onto the $(R_0, R_1)$ plane of that intersection, we apply Lemma \ref{Lemma_FME_R1_R0} to each constituent polytope and then take the intersection of the $L$ resulting polygons.

For each $s \in \mathcal{S}$, we project the polytope described by (\ref{Error_common_Rec}), (\ref{Error_private_Rec}), (\ref{Eq_sequence10_3}) and (\ref{Eq_sequence10_4}) onto the two-dimensional plane ($R_0,R_1$), using Lemma \ref{Lemma_FME_R1_R0} (by setting $U=X$, $Y=Y_s$ and $k=L+1$ therein) to obtain for each $ s\in \mathcal{S} $,
\begin{align}
&R_{0}{+} R_1 {\leq}I(X;Y_s|U_{c}) {+} I(U_{c};Y_{c})\nonumber\\
&R_{0}{\leq}  I(U_{c};Y_{c}) \enspace c\in \{L{+}1,L{+}2,...,K{-}1\}
\label{Region_After_projection_Final}
\end{align}
Thus, the above inequalities for all $ s\in \mathcal{S} $ together with \eqref{Error_Simple_1}, \eqref{Error_simple_2}, and \eqref{Error_simple_3} describe the achievability region in (\ref{Region_Genral_Case}), completing the proof of Theorem \ref{TH_genral_case}.

\section{Proof of Lemma  \ref{Lemma_FME_R1k_R1l}}
\label{Appendix_proof_FME_structure}    

We prove Lemma \ref{Lemma_FME_R1k_R1l} by mathematical induction. 

We begin by showing the initial case of $l=k$ using FME to project out $R_{1k}$ from the set of inequalities in \eqref{Eq_structure0_1}-\eqref{Eq_structure0_4}. We rewrite these inequalities to emphasize those that contain $R_{1k}$ as follows 
\begin{align}
&R_{0}{+}R_{1k}{+} \sum_{j{=}k+1}^{K}R_{1j}\leq  I(U_{k};Y_{k}) \label{Eq_structure13_0}\\
&R_{0}{+} \sum_{j{=}c}^{K}R_{1j}\leq  I(U_{c};Y_{c}) \enspace \enspace c\in\{k{+}1,...,K\} \label{Eq_structure13_1}\\
&R_{1} {-} R_{1k}{-}\sum_{j{=}k+1}^{K}R_{1j}\leq  I(U;Y|U_{k})\label{Eq_structure13_2}\\
&R_{1} {-} \sum_{j{=}c}^{K}R_{1j}\leq  I(U;Y|U_{c})\enspace \enspace c\in\{k{+}1,...,K\}\label{Eq_structure13_3}\\
&-R_1{+}R_{1k}{+}\sum_{j=k+1}^K R_{1j}\leq 0 \label{Eq_structure13_4}\\
& -R_{1k}\leq 0 \label{Eq_structure13_5} \\
& -R_{1c}\leq 0 \enspace \enspace c\in\{k+1,...,K\} \label{Eq_structure13_6}
\end{align} 
In the the above inequalities, only \eqref{Eq_structure13_0},\eqref{Eq_structure13_2}, \eqref{Eq_structure13_4}, and \eqref{Eq_structure13_5} contain $R_{1k}$. 

By eliminating $R_{1k}$ from \eqref{Eq_structure13_0} and \eqref{Eq_structure13_2}, we have 
\begin{equation}
\label{Eq:FME_R_1k_1}
R_0+R_1\leq I(U;Y|U_k)+I(U_k;Y_k)
\end{equation}
and from \eqref{Eq_structure13_0} and \eqref{Eq_structure13_5} we get 
\begin{equation}
\label{Eq:FME_R_1k_2}
R_{0}{+} \sum_{j{=}k+1}^{K}R_{1j}\leq  I(U_{k};Y_{k})
\end{equation}
Finally, we project away $R_{1k}$ from \eqref{Eq_structure13_4} and \eqref{Eq_structure13_5} to get
\begin{equation}
\label{Eq:FME_R_1k_3}
-R_1{+}\sum_{j=k+1}^K R_{1j}\leq 0
\end{equation}
Now consider \eqref{Eq:FME_R_1k_1}-\eqref{Eq:FME_R_1k_3} together with \eqref{Eq_structure13_1}, \eqref{Eq_structure13_3} and \eqref{Eq_structure13_6}. These inequalities together form the polytope defined by \eqref{Eq_structure0_1}-\eqref{Eq_structure0_4} with $l=k$. Thus the lemma is true for $l=k$.


Next, consider the induction hypothesis that the lemma is true for $l=i$ for some  $ K=1 \leq i \leq K-2 $. 
In other words, by projecting ways the split rates $R_{1k},..,R_{1i}$, we have
\begin{align}
&R_0{+}R_1  {\leq} I(U;Y|U_c){+}I(U_c;Y_c) \enspace c{\in} \{k,k{+}1,..,i\} \label{Eq_structure10_1}\\
& R_0{+}\sum_{j=i{+}1}^K R_{1j}\leq I(U_c;Y_c) \enspace c\in \{k,k{+}1,..,i\}\label{Eq_structure10_2}\\
&R_{0}{+} \sum_{j{=}c}^{K}R_{1j}\leq  I(U_{c};Y_{c}) \enspace \enspace c\in\{i{+}1,i{+}2.,..,K\} \label{Eq_structure10_3}\\
&R_{1} {-} \sum_{j{=}c}^{K}R_{1j}\leq  I(U;Y|U_{c})\enspace \enspace c\in\{i{+}1,i{+}2,...,K\}\label{Eq_structure10_4}\\
&-R_1{+}\sum_{j=i{+}1}^K R_{1j}\leq 0 \label{Eq_structure10_5}\\
&-R_{1c}\leq 0 \enspace \enspace c\in\{i{+}1,...,K\} \label{Eq_structure10_6}
\end{align}

Next, we show that eliminating $R_{1i+1}$ from the above polytope, we obtain a polytope that is given by Lemma \ref{Lemma_FME_R1k_R1l} with $l=i+1$, thus completing its proof. 

Let us rewrite the above sets of inequalities to illustrate the dependence on $R_{1i+1}$ as follows:

\begin{align}
&R_0+R_1  {\leq} I(U;Y|U_c){+}I(U_c;Y_c) \enspace c{\in} \{k,k{+}1,..,i\} \label{Eq_structure11_1}\\
& R_0{+}R_{1i+1}{+}\sum_{j=i{+}2}^K R_{1j}\leq I(U_c;Y_c) \enspace c{\in} \{k,k{+}1,..,i\}\label{Eq_structure11_2}\\
&R_{0}{+} \sum_{j{=}c}^{K}R_{1j}\leq  I(U_{c};Y_{c}) \enspace \enspace c\in\{i{+}2.,..,K\} \label{Eq_structure11_3}\\
&R_{0}{+} R_{1i+1}{+}\sum_{j{=}i+2}^{K}R_{1j}\leq  I(U_{i+1};Y_{i+1}) \label{Eq_structure11_4}\\
&R_{1} {-} \sum_{j{=}c}^{K}R_{1j}\leq  I(U;Y|U_{c})\enspace \enspace c\in\{i{+}2,...,K\}\label{Eq_structure11_5}\\
&R_{1} {-} R_{1i+1}{-}\sum_{j{=}i+2}^{K}R_{1j}\leq  I(U;Y|U_{i+1})\label{Eq_structure11_6}\\
&-R_1{+}R_{1i+1}{+}\sum_{j=i{+}2}^K R_{1j}\leq 0 \label{Eq_structure11_7}\\
&-R_{1c}\leq 0 \enspace \enspace c\in\{i{+}2,...,K\} \label{Eq_structure11_8}\\
&-R_{1i+1}\leq 0 \label{Eq_structure11_9}
\end{align}
The inequalities \eqref{Eq_structure11_2}, \eqref{Eq_structure11_4} and \eqref{Eq_structure11_7}, give upper bounds on $R_{1i+1}$ while  \eqref{Eq_structure11_6} and \eqref{Eq_structure11_9} give lower bounds on $R_{1i+1}$. 

By eliminating $R_{1i+1}$ from \eqref{Eq_structure11_2} and \eqref{Eq_structure11_6}, we have 
\begin{equation}
\label{Eq_FME1}
R_0{+}R_1{\leq} I(U;Y|U_{i+1}){+}I(U_c;Y_c) \enspace c{\in}\{k,k+1,..,i\}
\end{equation} 
From \eqref{Eq_structure11_1}, all the inequalities in (\ref{Eq_FME1}) are redundant since $I(U;Y|U_{i+1}) \geq I(U;Y|U_{c}) \enspace \forall c\in \{k,k+1,...,i\}$. This follows form the Markov chain $U_{K}\markov U_{K-1}\markov \cdots \markov U_{i+1}\markov U_{i}\markov \cdots \markov U_{k} \markov U\markov (Y_K,Y_{K-1},..,Y_{i+1},Y_i,..,Y_k,Y)$.

Then, by projecting $R_{1i+1}$ from \eqref{Eq_structure11_2} and \eqref{Eq_structure11_9}, we have 
\begin{equation}
\label{Eq_FME2}
R_0{+}\sum_{j=i{+}2}^K R_{1j}\leq I(U_c;Y_c) \enspace c\in \{k,k+1,..,i\}\end{equation}
From \eqref{Eq_structure11_4} and \eqref{Eq_structure11_6}, we also eliminate $R_{1i+1}$ to get
\begin{equation}
\label{Eq_FME3}
R_0{+}R_1\leq I(U;Y|U_{i+1})+I(U_{i+1};Y_{i+1})
\end{equation}
and from \eqref{Eq_structure11_4} and \eqref{Eq_structure11_9}, to get
\begin{equation}
\label{Eq_FME4}
R_0{+}\sum_{j{=}i+2}^{K}R_{1j}\leq I(U_{i+1};Y_{i+1})
\end{equation}
Finally, we eliminate $R_{1i+1}$ from \eqref{Eq_structure11_7} and \eqref{Eq_structure11_6} to get the redundant inequality $0 \leq I(U;Y|U_{i+1})$ and from \eqref{Eq_structure11_7} and \eqref{Eq_structure11_9} to get
\begin{equation}
\label{Eq_FME5}
-R_1{+}\sum_{j=i{+}2}^K R_{1j}\leq 0
\end{equation}

Hence, after projecting away the split rates $R_{1k},...,R_{1i},R_{i+1}$, we get the  polytope 
\begin{align}
&R_0{+}R_1  {\leq} I(U;Y|U_c){+}I(U_c;Y_c) \enspace c{\in} \{k,..,i+1\} \label{Eq_structure12_1}\\
& R_0{+}\sum_{j=i{+}2}^K R_{1j}\leq I(U_c;Y_c) \enspace c\in \{k,k+1,..,i+1\}\label{Eq_structure12_2}\\
&R_{0}{+} \sum_{j{=}c}^{K}R_{1j}\leq  I(U_{c};Y_{c}) \enspace \enspace c\in\{i+2,i+3.,..,K\} \label{Eq_structure12_3}\\
&R_{1} {-} \sum_{j{=}c}^{K}R_{1j}\leq  I(U;Y|U_{c})\enspace \enspace c\in\{i+2,i+3,...,K\}\label{Eq_structure12_4}\\
&-R_1{+}\sum_{j=i{+}2}^K R_{1j}\leq 0 \label{Eq_structure12_5}\\
&-R_{1c}\leq 0 \enspace \enspace c\in\{i{+}2,...,K\} \label{Eq_structure12_6}
\end{align}
which coincides with that described by (\eqref{Eq_structure5_1}-\eqref{Eq_structure5_6}) in Lemma \ref{Lemma_FME_R1k_R1l} for $l{=}i{+}1$, as was to be proved.  Hence, by the principle of mathematical induction, Lemma \ref{Lemma_FME_R1k_R1l} is true for any $l{\in}\{k,k+1,..,K{-}1\}$, thus concluding its proof.

\section{Is Non-Unique decoding Necessary?}
\label{app:nudnec?}
At this point, we consider the following question: is the improvement in Corollary \ref{Corollary_OneIndirect} over Case (i) in Theorem \ref{Th_Ext_KM_region} due to rate splitting, indirect decoding, or both? In the next proposition, we show that when indirect decoding is replaced with direct decoding there is no loss of achievable rate region. However, the role of rate splitting appears to be essential for Corollary \ref{Corollary_OneIndirect}.

\begin{proposition}
\label{Prop_No_Need_IndirectDecoding}
Indirect decoding is unnecessary in Corollary \ref{Corollary_OneIndirect}. In other words, the result of Corollary \ref{Corollary_OneIndirect} can also be obtained if the common receiver $Y_{L+1}$ uses joint (unique) decoding instead of indirect decoding. 
\end{proposition}
\begin{proof}
If $Y_{L{+}1}$ uses joint decoding instead of indirect decoding, we get the following achievable rate region 
\begin{align*}
&R_0\leq \min \{I(U_K;Y_c),I(U_{L{+}1};Y_{L{+}1})\} \enspace \forall c{\in} \mathcal{C} \backslash \{L+1 \} \nonumber \\
&R_1\leq I(X;Y_r|U_K)\nonumber \\
&R_1\leq I(X;Y_r|U_{L{+}1}){+}I(U_{L{+}1};Y_{L{+}1}|U_K)\\
&R_0{+}R_1\leq I(X;Y_r|U_{L{+}1}){+}I(U_{L{+}1};Y_{L{+}1})
\end{align*}
for some $U_K{\markov} U_{L{+}1}{\markov} X$. We refer to the above region as the joint decoding region. Recall there $r$ is such that $Y_r \preceq Y_s$ for all $s\in \mathcal{S} \backslash \{r\} $. We show that the extra inequality (the third one) is redundant by optimizing the choice of $U_{L{+}1}$ given $U_K$. We have the following two cases.
\begin{enumerate}
\item $I(U_K;Y_c)\geq I(U_K;Y_{L{+}1}) \enspace  \forall c{\in} \{L+2,...,K\}$: In this case, for any $U_{L{+}1}$, $R_1$ at the corner point of the indirect region can be described as
\begin{align}
R_1^*&{=} \min\{I(X;Y_r|U_K),I(X;Y_r|U_{L{+}1}){+}\nonumber \\&I(U_{L{+}1};Y_{L{+}1}){-} \min \{I(U_K;Y_c),I(U_{L{+}1},Y_{L{+}1})\}\} \nonumber\\
&{=}I(X;Y_r|U_{L{+}1}){+}\min\{I(U_{L{+}1};Y_r|U_K), \nonumber \\
&I(U_{L{+}1};Y_{L{+}1}){-} \min \{I(U_K;Y_c),I(U_{L{+}1};Y_{L{+}1})\}\} \label{markov_chain11}\\ 
&{=}I(X;Y_r|U_{L{+}1}){+}\min\{I(U_{L{+}1};Y_r|U_K),\nonumber \\
&I(U_{L{+}1};Y_{L{+}1}|U_K){+}I(U_K;Y_{L{+}1}){-}\nonumber \\ 
&\min \{I(U_K;Y_c),I(U_{L{+}1};Y_{L{+}1})\}\} \label{markov_chain22}
\end{align} where (\ref{markov_chain11}) and (\ref{markov_chain22}) follows from the Markov chain $U_K\markov U_{L+1}\markov X$ and the chain rule of the mutual information. 
In this case, \begin{equation}
\min \{I(U_K;Y_c),I(U_{L{+}1};Y_{L{+}1})\}\leq I(U_K;Y_{L{+}1})\nonumber
\end{equation}
From the above inequality, it is clear that
\begin{equation}
R_1^*{\leq}I(X;Y_r|U_{L{+}1})+I(U_{L{+}1};Y_{L{+}1}|U_K)\nonumber
\end{equation}
Hence, the extra inequality in the joint decoding region is redundant for any $U_{L{+}1}$ given $U_K$.
\item $I(U_K;Y_c)\leq  I(U_K;Y_{L{+}1})$  for any $c{\in} C$ and $C$ is any subset $C \subseteq \{L+2,...,K\}$:  equivalently, we have at least one mutual information term $I(U_K;Y_c)\enspace c\in\{L+2,...,K\}$ is less than $I(U_K;Y_{L{+}1})$. Without loss of generality, we can assume that
\begin{align*}
\min \{I(U_K;Y_c)\}=I(U_K;Y_t) \leq  I(U_K;Y_{L{+}1}) \\ \text{where},
\enspace c\in\{L+2,...,K\}
\end{align*}
The index $t$ exists, since we assumed that we have at least one mutual information term that is less than $I(U_K;Y_{L{+}1}) $. For this case, $R_1$ at the corner point of the indirect region can be written as
\begin{align*}
R_1^*&{=}\min\{I(X;Y_r|U_K),I(X;Y_r|U_{L{+}1}){+}\\&I(U_{L{+}1};Y_{L{+}1}){-}I(U_K;Y_t)
\end{align*}
Hence, we need to choose $U_{L+1}$ given $U_K$ such that 
\begin{align*}
I(X;Y_r|U_K){\leq}I(X;Y_r|U_{L{+}1}){+}&I(U_{L{+}1};Y_{L{+}1}){-}\\&I(U_K;Y_t)
\end{align*}

The optimal choice is to set $U_{L{+}1}{=}U_K$ because $I(U_{L{+}1};Y_{L{+}1}){\geq}I(U_K;Y_t)$. For this choice, indirect decoding and joint decoding again yield the same region. 

\end{enumerate}
\end{proof}


\begin{remark}
The work in \cite{Romero-Varanasi:isit2017} provides a general achievable rate region for the $K$-user DM BC with a general message set (any subset of the set of the exponentially many possible messages, one for each subset of receivers). That inner bound is based on a general form of rate-splitting, superposition coding and partial interference (unique) decoding.
In particular, in \cite{Romero-Varanasi:isit2017}, it was shown that known inner bounds/capacity regions in 2-receiver and 3-receiver BC special cases based on rate-splitting and superposition coding are reproduced by the general inner bound therein. Corollary \ref{Corollary_OneIndirect} and Proposition \ref{Prop_No_Need_IndirectDecoding} imply that the general inner bound in \cite{Romero-Varanasi:isit2017} when specialized to the $K$-user DM BC with two nested multicast messages and the particular rate-splitting strategy of Corollary \ref{Corollary_OneIndirect} along with unique decoding as in Proposition \ref{Prop_No_Need_IndirectDecoding} coincides with the capacity region for the new class of DM BCs specified by Corollary \ref{Corollary_OneIndirect}.
\end{remark}

\end{appendices}

\bibliographystyle
{IEEEtran}
\bibliography{IEEEabrv,Cite}

\begin{thebibliography}{10}
\providecommand{\url}[1]{#1}
\csname url@samestyle\endcsname
\providecommand{\newblock}{\relax}
\providecommand{\bibinfo}[2]{#2}
\providecommand{\BIBentrySTDinterwordspacing}{\spaceskip=0pt\relax}
\providecommand{\BIBentryALTinterwordstretchfactor}{4}
\providecommand{\BIBentryALTinterwordspacing}{\spaceskip=\fontdimen2\font plus
\BIBentryALTinterwordstretchfactor\fontdimen3\font minus
  \fontdimen4\font\relax}
\providecommand{\BIBforeignlanguage}[2]{{%
\expandafter\ifx\csname l@#1\endcsname\relax
\typeout{** WARNING: IEEEtran.bst: No hyphenation pattern has been}%
\typeout{** loaded for the language `#1'. Using the pattern for}%
\typeout{** the default language instead.}%
\else
\language=\csname l@#1\endcsname
\fi
#2}}
\providecommand{\BIBdecl}{\relax}
\BIBdecl

\bibitem{korner1975source}
J.~Korner and K.~Marton, ``A source network problem involving the comparison of
  two channels {II},'' \emph{IEEE Transactions on Information Theory}, 1975.

\bibitem{kiirner1977general}
------, ``General broadcast channels with degraded message sets,'' \emph{IEEE
  Transactions on Information Theory}, vol.~23, pp. 60--64, 1977.

\bibitem{gallager1974capacity}
R.~G. Gallager, ``Capacity and coding for degraded broadcast channels,''
  \emph{Problemy Peredachi Informatsii}, vol.~10, no.~3, pp. 3--14, 1974.

\bibitem{korner1977images}
J.~Korner and K.~Marton, ``Images of a set via two channels and their role in
  multi-user communication,'' \emph{IEEE Transactions on Information Theory},
  vol.~23, no.~6, pp. 751--761, 1977.

\bibitem{diggavi2006opportunistic}
S.~N. Diggavi and D.~Tse, ``On opportunistic codes and broadcast codes with
  degraded message sets,'' in \emph{2006 IEEE Information Theory
  Workshop-ITW'06 Punta del Este}.\hskip 1em plus 0.5em minus 0.4em\relax IEEE,
  2006, pp. 227--231.

\bibitem{nair2009capacity}
C.~Nair and A.~El~Gamal, ``The capacity region of a class of three-receiver
  broadcast channels with degraded message sets,'' \emph{IEEE Transactions on
  Information Theory}, vol.~55, no.~10, pp. 4479--4493, 2009.

\bibitem{csiszar1978broadcast}
I.~Csisz{\'a}r and J.~Korner, ``Broadcast channels with confidential
  messages,'' \emph{IEEE Transactions on Information Theory}, vol.~24, no.~3,
  pp. 339--348, 1978.

\bibitem{ngai2004network}
C.~K. Ngai and R.~W. Yeung, ``Network coding gain of combination networks,'' in
  \emph{IEEE Information Theory Workshop, 2004}.\hskip 1em plus 0.5em minus
  0.4em\relax IEEE, 2004, pp. 283--287.

\bibitem{bidokhti2016capacity}
S.~S. Bidokhti, V.~M. Prabhakaran, and S.~N. Diggavi, ``Capacity results for
  multicasting nested message sets over combination networks,'' \emph{IEEE
  Transactions on Information Theory}, vol.~62, no.~9, pp. 4968--4992, 2016.

\bibitem{Romero-Varanasi:isit2016}
H.~P. Romero and M.~K. Varanasi, ``Superposition coding in the combination
  network,'' in \emph{IEEE International Symposium on Information Theory
  Proceedings ({ISIT})}.\hskip 1em plus 0.5em minus 0.4em\relax IEEE, 2016.

\bibitem{Romero-Varanasi:isit2017}
------, ``Rate-splitting and superposition coding for concurrent groupcasting
  over the broadcast channel: a general framework,'' in \emph{IEEE
  International Symposium on Information Theory Proceedings ({ISIT})}.\hskip
  1em plus 0.5em minus 0.4em\relax IEEE, 2017.

\bibitem{nair2011capacity}
C.~Nair and Z.~V. Wang, ``The capacity region of the three receiver less noisy
  broadcast channel,'' \emph{IEEE Transactions on Information Theory}, vol.~57,
  no.~7, pp. 4058--4062, 2011.

\bibitem{schrijver1986theory}
A.~Schrijver, ``Theory of integer and linear programming,'' 1986.

\bibitem{nair2012three}
C.~Nair and L.~Xia, ``On three-receiver more capable channels,'' in \emph{IEEE
  International Symposium on Information Theory Proceedings (ISIT) 2012}.\hskip
  1em plus 0.5em minus 0.4em\relax IEEE, 2012, pp. 378--382.

\bibitem{Salman2018Ach}
M.~Salman and M.~K. Varanasi, ``An achievable rate region for the {K}-receiver
  two nested groupcast {DM} broadcast channel and a capacity result for the
  combination network,'' in \emph{IEEE International Symposium on Information
  Theory Proceedings ({ISIT})}.\hskip 1em plus 0.5em minus 0.4em\relax IEEE,
  2018.

\bibitem{chong2008han}
H.-F. Chong, M.~Motani, H.~K. Garg, and H.~El~Gamal, ``On the han--kobayashi
  region for theinterference channel,'' \emph{IEEE Transactions on Information
  Theory}, vol.~54, no.~7, pp. 3188--3195, 2008.

\bibitem{han1981new}
T.~Han and K.~Kobayashi, ``A new achievable rate region for the interference
  channel,'' \emph{IEEE transactions on information theory}, vol.~27, no.~1,
  pp. 49--60, 1981.

\bibitem{bidokhti2014non}
S.~S. Bidokhti and V.~M. Prabhakaran, ``Is non-unique decoding necessary?''
  \emph{IEEE Transactions on Information Theory}, vol.~60, no.~5, pp.
  2594--2610, 2014.

\end{thebibliography}

\end{document}